\documentclass[10pt,leqno,oneside]{extarticle}

\usepackage{ams math,amssymb,amsfonts,amsthm,mathrsfs,MnSymbol}
\usepackage{cite}

\usepackage[small]{caption}
\usepackage{amsmath}
\usepackage{amssymb}
\usepackage{amsfonts}
\usepackage{setspace}
\usepackage{latexsym}
\usepackage{mathrsfs}
\usepackage{epsfig}
\usepackage{amsthm}
\usepackage{yfonts}
\usepackage{graphicx}
\usepackage{color}

\usepackage{enumitem}
\setitemize{wide}

\newcommand{\n}{\noindent}
\newcommand{\be}{\begin{equation}}
\newcommand{\ee}{\end{equation}}
\newcommand{\ben}{\begin{displaymath}}

\newcommand{\een}{\end{displaymath}}

\newcommand{\vs}{\vspace{0.2cm}}

\newcommand{\scb}{\scalebox}

\newtheorem{Proposition}{Proposition}
\newtheorem{Theorem}{Theorem}

\newtheorem{Lemma}{Lemma}

\newtheorem{Corollary}{Corollary}

\newtheorem{Note}{Note}

\usepackage{etoolbox}
\begin{document}

\newcommand{\north}{{\sf N}}
\newcommand{\south}{{\sf S}}
\newcommand{\poles}{{\mathcal{P}}}
\newcommand{\mc}{\text{tr}_{h}\Theta}

\newcommand{\meanc}{{\rm tr}_{h}\Theta}
\newcommand{\narea}{\frac{A(S)}{4\pi}}
\newcommand{\nareat}{\frac{A(S_{t})}{4\pi}}

\newcommand{\rc}{\color{red}}

\begin{center}
{\huge On extreme Kerr-throats and zero 

\vs
temperature black-holes.}

\vspace{0.6cm}
{\sc Martin Reiris.}\\

\vs
\textsc{Max Planck Institute f\"ur Gravitationsphysik. \\ Albert Einstein Institut - Germany.}\\e-mail: martin@aei.mpg.de.

\vspace{0.6cm}
\begin{minipage}[c]{12cm}
\begin{spacing}{1}
{\small Recently it was shown that the area $A$ and the angular momentum $J$ of any apparent horizon on a maximal, axisymmetric and asymptotically flat Cauchy hyper-surface of a vacuum space-time satisfy necessarily the universal inequality $A\geq 8\pi |J|$. 
We show here that the equality $A=8\pi |J|$ is never attained. 
%
We study too the global structure of data sets having surfaces with $A=8\pi |J|$. This lead us to prove the rigidity of the extreme Kerr-throats and to investigate the important phenomenon of formation of extreme Kerr-throats along sequences of data sets.}
\end{spacing}
\end{minipage}

\end{center}


\section{Introduction.}

The celebrated Penrose singularity theorem asserts (in particular) that
when a trapped surface is present in an asymptotically flat Cauchy hyper-surface of a given globally hyperbolic vacuum space-time then such space-time is necessarily null geodesically incomplete \cite{MR757180}. Because of this and other facts, trapped surfaces are usually associated to the emergence of black holes and are therefore central objects of study in General Relativity. 
Keeping this in mind let us concentrate in axisymmetric and asymptotically flat vacuum space-times and in maximal axisymmetric Cauchy hyper-surfaces. Moreover suppose that over the Cauchy hyper-surface there is a trapped surface and that it is isotopic to the sphere at ``infinity" over one of the possibly many ends as is depicted in Figure \ref{Fig1}. In this scenario the boundary of the trapped region in the hyper-surface is known to be a stable Marginally Outer Trapped Surface (MOTS), called the {\it apparent horizon}, which in turn is usually interpreted as a quasi-localization of the event horizon \cite{MR2884392}. On the other hand in axisymmetry every (embedded, orientable, compact and boundary-less [\footnote{These will be common assumptions.}]) surface has associated its Komar angular momentum $J(S)$ that depends only on the homology class of the surface (see Section \ref{ANGULARMOMEN}). In particular the angular momentum of the apparent horizon coincides with that of the respective asymptotically flat end. Moreover
as was shown in \cite{2011PhRvL.107e1101D} the universal inequality 
\be\label{AJI}
A(S)\geq  8\pi |J(S)|,
\ee
holds between the area $A(S)$ of any embedded surface $S$ and its angular momentum $J(S)$.
It is concluded then that the area of the apparent horizon is always greater or equal than $8\pi$ times the angular momentum of the respective asymptotically flat end. 
In this article we show that apparent horizons saturating the inequality (\ref{AJI}), namely  with $A=8\pi |J|$, cannot exist in maximal, axisymmetric and asymptotically flat vacuum data sets (Theorem \ref{MLemma} and Corollary \ref{Coroll1}). 
But we also investigate what occurs to the geometry of this type of data sets around apparent horizons nearly saturating (\ref{AJI}) (Corollary \ref{Cor2}). This is closely related to analyzing the global structure of data sets (of a different global type) admitting a surface saturating (\ref{AJI}) and that we investigate in Theorem \ref{Lemma3}. As we will see the quest has deep theoretical implications. We will  be explaining all this in full detail in the discussion below.

Most of the discussions in this article are centered around the notions of extreme Kerr-throat and extreme Kerr-throat sphere. 
To begin explaining these notions let us consider the family of the Kerr-solutions in the Boyer-Lindquist coordinates 
\begin{align}
\label{Km} {\bf g}=& -\bigg[\frac{\Delta - a^{2}\sin^{2}\theta}{\Sigma}\bigg]\ dt^{2} - 
\frac{a\sin^{2}\theta(r^{2}+a^{2}-\Delta)}{\Sigma}(dt d\varphi+d\varphi dt)\\
\nonumber & +\bigg[\frac{(r^{2}+a^{2})^{2}-\Delta a^{2} \sin^{2}\theta}{\Sigma}\bigg]\sin^{2}\theta\ d\varphi^{2}+
\frac{\Sigma}{\Delta}\ dr^{2}+\Sigma\ d\theta^{2},
\end{align}
where here
$\Sigma=r^{2}+a^{2}\cos^{2}\theta,\ 
\Delta=r^{2}+a^{2}-2m r,
$
and $a=J/m$. The family is parametrized by the mass $m$ and the angular momentum $J$. The Kerr black holes correspond to the range of parameters $m^{2}\geq |J|$ while the extreme Kerr-black holes correspond to the case when $m^{2}=|J|$ (i.e. $a^{2}=m^{2}$). 
Note that in this last case we have $\Delta=(r-m)^{2},\ r>m$. 
Let us restrict the attention to the maximal slice $\{t=0\}$. Over this slice the extreme Kerr solutions have 
one asymptotically flat end (when $r\uparrow \infty$) and one cylindrical end (when $r\downarrow m$). These ends are depicted in Figure \ref{Fig2}. The cylinder possesses asymptotically a well 
defined smooth (i.e. $C^{\infty}$) data $(S^{2}\times \mathbb{R};g_{T},K_{T})$ indexed here with a $T$ from ``Throat" and called the {\it extreme kerr-throat of angular momentum $|J|$} (or of mass $m=\sqrt{|J|}$). The explicit form of the data is
\begin{align}
& \label{mKt} g_{T}=\bigg(\frac{4|J|\sin^{2}\theta}{1+\cos^{2}\theta}\bigg)\, d\varphi^{2}+|J|(1+\cos^{2}\theta)\, d\theta^{2} +|J| (1+\cos^{2}\theta)\, d\tilde{r}^{2},\\
& \label{KKt} K_{T}=\bigg(\frac{2\sqrt{|J|}\sin^{2}\theta}{(1+\cos^{2}\theta)^{\frac{3}{2}}}\bigg)\, (d\varphi d\tilde{r}+d\tilde{r}d\varphi).
\end{align}

\n The vector field $\partial_{\tilde{r}}$ is a Killing field and $|\partial_{\tilde{r}}|^{2}=|J|(1+\cos^{2}\theta)$. Therefore $\partial_{\tilde{r}}=\alpha_{T} \varsigma$ where $\alpha_{T}:=\sqrt{|J|(1+\cos^{2}\theta)}$ and $\varsigma$ is a unit field normal to the spheres of constant $\tilde{r}$. The function $\alpha_{T}$, which depends on $\sqrt{|J|}$, will play a fundamental role later on. Note that the data of the extreme Kerr-throats are parametrized by their angular momentum $|J|$ which plays the role of a global scale factor. 
The spheres of constant $\tilde{r}$ are called {\it extreme Kerr-throat spheres} of area $A=8\pi|J|$ (or of mass $m=\sqrt{A/8\pi}$). They are totally geodesic, i.e. as surfaces in $(\Sigma;g)$ have zero second fundamental form, and are in particular minimal. Moreover they are stable minimal and the second variation of the area is non-negative and zero if and only if it is in a direction proportional to $\partial_{\tilde{r}}=\alpha_{T}\varsigma$ (i.e. $\partial_{\bar{r}}$ up to a non-zero factor). The induced two-metric is
\be\label{INDUMETRIC}
h_{T}=\bigg(\frac{4|J|\sin^{2}\theta}{1+\cos^{2}\theta}\bigg)\, d\varphi^{2}+|J|(1+\cos^{2}\theta)d\theta^{2}, 
\ee
and in the basis $\{\partial_{\theta},\partial_{\varphi},\varsigma\}$ the only non-zero components of $K_{T}$ are
\be\label{INDUSFF}
K_{T}(\partial_{\varphi},\varsigma)=K_{T}(\varsigma,\partial_{\varphi})=\frac{2 \sin^{2}\theta}{(1+\cos^{2}\theta)^{2}}.
\ee
In this setup, an axisymmetric sphere $S$ of area $A(S)=8\pi |J(S)|$ embedded in a data $(\Sigma;g,K)$ is said to be an extreme Kerr-throat sphere if (i) it is totally geodesic in $(\Sigma;g)$, (ii) the induced metric expressed in areal coordinates $(\theta,\varphi)$ (see Section \ref{APCO}) is given by (\ref{INDUMETRIC})
and if (iii) in the basis $\{\partial_{\theta},\partial_{\varphi},\varsigma\}$, where $\varsigma$ is a unit normal to $S_{H}$ in $\Sigma$, the only non-zero components of $K$ are $K(\partial_{\varphi},\varsigma)=K(\varsigma,\partial_{\varphi})$ and given by (\ref{INDUSFF}).

\begin{figure}[h]
\centering
\includegraphics[width=12cm,height=8cm]{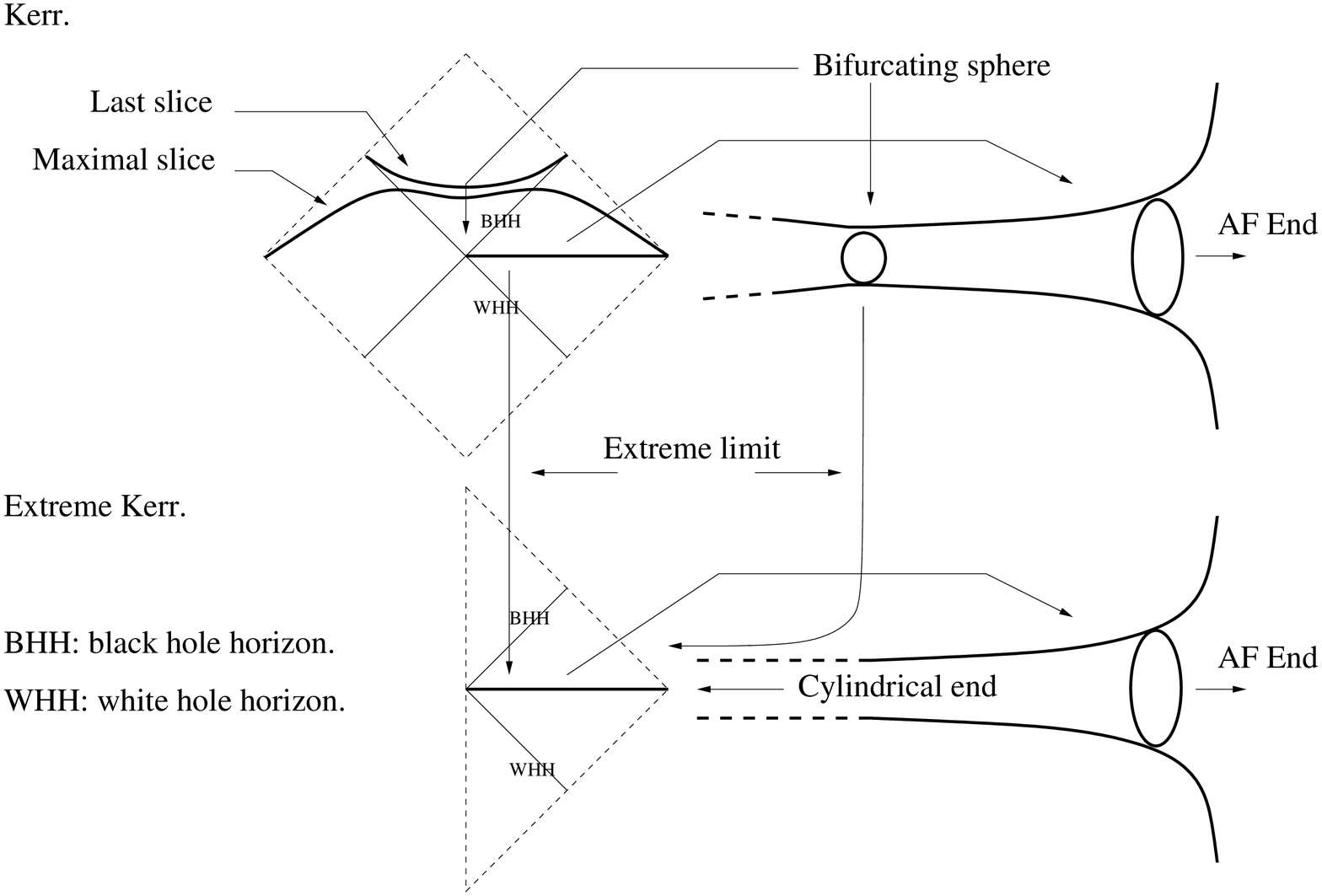}
\caption{}
\label{Fig2}
\end{figure} 

A fundamental result proved in \cite{2011PhRvL.107e1101D} and which is the basis to prove the universal inequality (\ref{AJI}) says that any stable and axisymmetric minimal surface with $A=8\pi |J|$ and embedded in a maximal and axisymmetric vacuum data set $(\Sigma;g,K)$ is necessarily an extreme Kerr-throat sphere regardless of the global nature of the data set in which it is embedded like completeness or asymptotic flatness. We will use this result very often.

Observe that if the equality in (\ref{AJI}) were reached in some surface $S$ inside an axisymmetric, maximal and asymptotically flat vacuum data set then such surface would have to be minimal and stable because the angular momentum of a surface depends only on its homology class [\footnote{More explicitly, for any smooth $F:[-\varepsilon,\varepsilon]\times S\rightarrow \Sigma$ with $F(0,-)={\rm Id} (-)$ and $\varepsilon$ small to have $F(x,-):S^{2}\rightarrow \Sigma$ a smooth embedding, the real function $\lambda\rightarrow A(F(\lambda,S))$ must have an absolute minimum at $\lambda=0$ because of (\ref{AJI}). It follows that the first $\lambda$-derivative is zero and the second is non-negative. As this is valid for all $F$ then the surface is minimal and stable.}]. As a result such surface would have to be an extreme Kerr-throat sphere.  

There is a worth mentioning interpretation for the surfaces $S$ saturating (\ref{AJI}) in terms of the well known thermodynamical heuristic from which we borrowed the terminology ``{\it zero temperature black hole}" that we used in the title. To explain this observe that non extremal Kerr black holes have a horizon in the slice $\{t=0\}$ of positive ``temperature"
\be\label{TEMP}
{\sf T}=\frac{1-\big(\frac{8\pi J}{A}\big)^{2}}{\sqrt{\frac{A}{4\pi} +\frac{16\pi J^{2}}{A}}},
\ee  
while the extremal Kerr black holes have  an ``asymptotic" horizon of zero ``temperature" because $A=8\pi |J|$.  If we import (\ref{TEMP}) as a raw definition for the temperature of an embedded surface then according to what was said before surfaces of zero temperature in vacuum axisymmetric and asymptotically flat data sets would be simply extreme Kerr-throat spheres of a particular area. 
In the context of the present discussion it is thus pertinent to ask whether zero temperature apparent horizons can arise in this type of data set or if on the contrary they can not, but could arise as in the extreme Kerr solutions only as asymptotic horizons on cylindrical ends of data sets. The present article investigates such situation and further related topics. 

Our first result, Theorem \ref{MLemma}, shows that indeed no surface exists saturating (\ref{AJI}) and embedded in an axisymmetric and asymptotically flat, maximal vacuum data set. 
\begin{Theorem}\label{MLemma} Let $(\Sigma;g,K)$ be a smooth vacuum axisymmetric maximal data set with finitely many asymptotically flat ends $E_{1},\ldots,E_{n}$. Let $S$ be any orientable compact and boundary-less embedded surface. Then 
\be\label{MI}
A(S)>8\pi |J(S)|,
\ee
where $A(S)$ is the area of $S$ and $J(S)$ is its angular momentum.
\end{Theorem}

An immediate corollary is 
\begin{Corollary}\label{Coroll1} There are no black hole apparent horizons of zero temperature in smooth, maximal, axisymmetric and asymptotically flat vacuum data sets.
\end{Corollary}
\begin{figure}[h]
\centering\includegraphics[width=10cm,height=5cm]{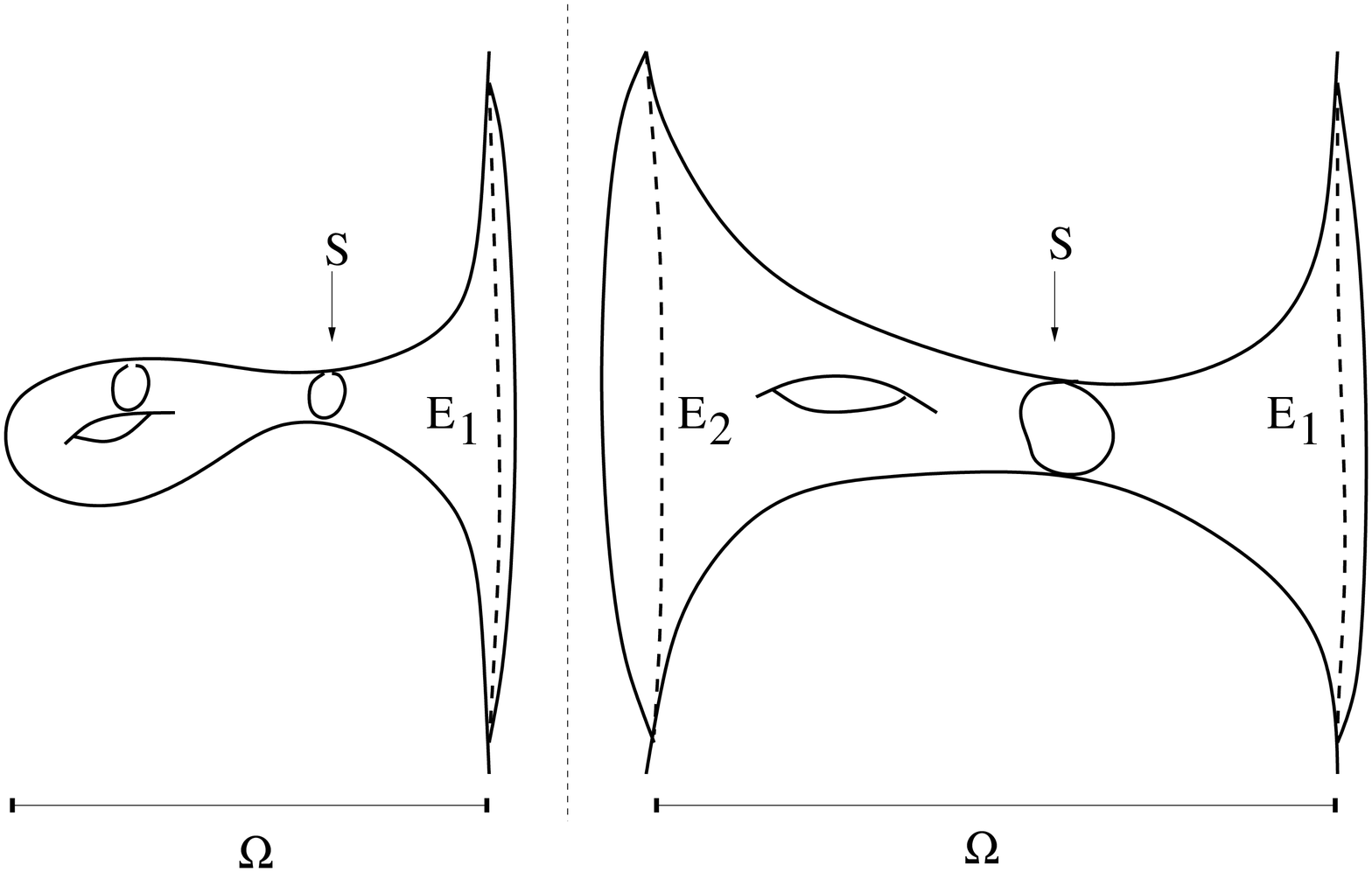}
\caption{Two possible configurations of topological black holes with the horizons marked with a $S$. On the left there is only one asymptotically flat end $E_{1}$, while on the right there are two, $E_{1}$ and $E_{2}$. Described are also large convex spheres on every asymptotically flat end. The horizon on the left has zero angular momentum because it encloses a compact boundary. The other surface shown on the left as well as the horizon shown on the right can have a priori non-zero angular momentum.}
\label{Fig1}\end{figure} 
Of course the extreme Kerr-throats are instances of maximal axisymmetric data sets possessing extreme Kerr-throat spheres but they are not asymptotically flat. The following theorem partially explains the special role that the extreme Kerr-throats play among the data sets containing an extreme Kerr-throat sphere.    
\begin{Theorem}\label{Lemma3} {\rm (Rigidity of extreme Kerr-throats)} Let $(\Sigma;g,K)$ be an homogeneously regular smooth axisymmetric maximal vacuum data set where $\Sigma$ is diffeomorphic to $S^{2}\times \mathbb{R}$. Suppose that for any sphere $S$ isotopic to the factor $S^{2}$ we have $A(S)\geq 8\pi |J|$ where $J=J(S)=J([S^{2}])$, and suppose that there is at least one sphere $S_{H}$ also isotopic to the factor $S^{2}$ with $A(S_{H})=8\pi |J|$. Then $(\Sigma;g,K)$ is the extreme Kerr-throat data set (\ref{mKt})-(\ref{KKt}) of angular momentum $|J|$ and $S_{H}$ is an extreme Kerr-throat sphere of area $A(S_{H})=8\pi |J|$.
\end{Theorem} 

Theorem \ref{Lemma3} is a remarkable manifestation, like in the positive mass theorem, of the non-linearity of the constraint equations. The notion of {\it homogeneously regular} manifold which is explained in Section \ref{HRM} essentially says that the metric is controlled in $C^{2}$ (on certain coordinates) on every metric ball of a uniform radius. We do not know at the moment if this condition can be removed and if only completeness of the data set is enough (it would be nice to answer this question).  It seems however feasible to prove an optimal version of Theorem \ref{Lemma3} prescinding of the topological condition $\Sigma\approx S^{2}\times \mathbb{R}$. 

A consequence of Theorem \ref{Lemma3} is the following important corollary on the formation of Kerr-throats.
\begin{Corollary}\label{Cor2} {\rm (Formation of extreme Kerr-throats)} Let $(\Sigma;g_{i},K_{i})$ be a sequence of smooth maximal vacuum axisymmetric and asymptotically flat initial data sets and let $S_{i}$ be a sequence of spheres embedded in $\Sigma_{i}$. Suppose that the sequence of data sets converges smoothly into a homogeneously regular maximal data set $(\Sigma_{\infty};g_{\infty},K_{\infty})$ and that the sequence of spheres converges to a sphere $S_{\infty}$. If $\Sigma_{\infty}$ is diffeomorphic to $S^{2}\times \mathbb{R}$ and $A(S_{\infty})=8\pi |J(S_{\infty})|$, then the limit data $(\Sigma_{\infty};g_{\infty},K_{\infty})$ is the extreme Kerr-throat of angular momentum $|J(S_{\infty})|$ and $S_{\infty}$ is an extreme Kerr-throat sphere of area $A(S_{\infty})$.
\end{Corollary}  
The precise notions of convergence involved in this statement are the following. The sequence $(\Sigma;g_{i},K_{i})$ converges smoothly to $(\Sigma_{\infty};g_{\infty},K_{\infty})$ iff there is a sequence of diffeomorphisms $\varphi_{i}:\Sigma_{\infty}\rightarrow \Sigma_{i}$ such that $\varphi_{i}^{*}\, g_{i}$ and $\varphi^{*}_{i}\, K_{i}$ converge in $C^{\infty}$ and over any open set of compact closure to $g_{\infty}$ and $K_{\infty}$ respectively. The spheres $S_{i}$ converge smoothly to $S_{\infty}$ iff there are diffeomorphisms $\phi_{i}:S^{2}\rightarrow S_{i}$ such that $\varphi_{i}^{-1}\circ \phi_{i}: S^{2}\rightarrow \Sigma_{\infty}$ converges in $C^{\infty}$ to a smooth embedding $S^{2}\rightarrow S_{\infty}\subset \Sigma_{\infty}$. 

Roughly speaking what the corollary says is that under basic assumptions a sequence of asymptotically flat, maximal axisymmetric data $(\Sigma_{i};g_{i},K_{i})$ and sequence of embedded spheres $S_{i}$ can asymptotically saturate the universal inequality (\ref{AJI}) only at the expense of the formation of an extreme Kerr-throat. More heuristically: extreme Kerr-throats form as the ``temperature decreases to zero". This phenomenon is depicted in Figure \ref{Figure2}. 
\begin{figure}[h]
\centering
\includegraphics[width=12cm,height=6cm]{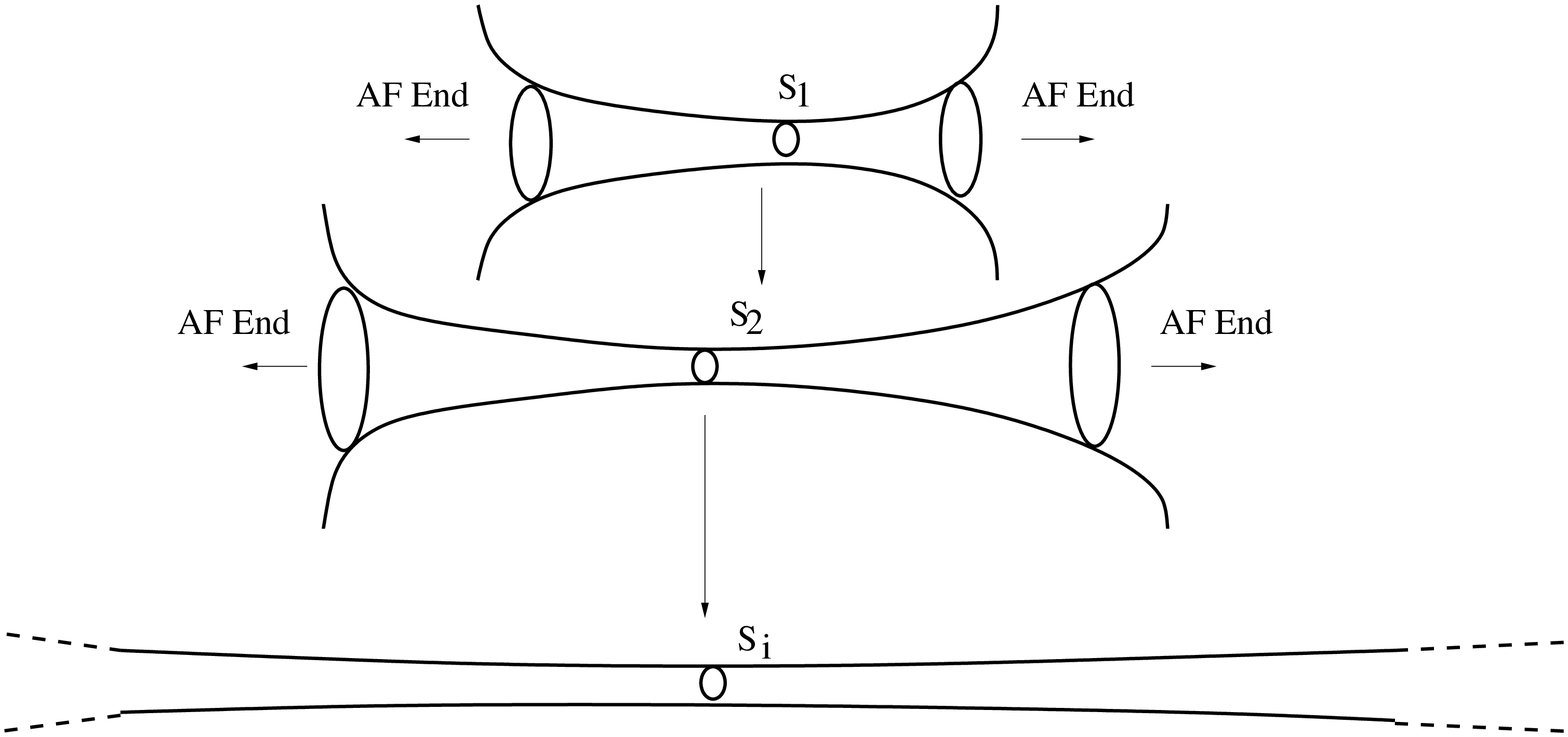}
\caption{The phenomenon of the formation of extreme Kerr-throats along sequence of data sets. The inequality (\ref{AJI}) is saturated asymptotically along the spheres $\{S_{i}\}$.}
\label{Figure2}
\end{figure} 

It is interesting to see Theorem \ref{MLemma} and Corollary \ref{Cor2} at work in the Kerr family of black holes. For the discussion that follows it is worth keeping in mind the Penrose diagram in Figure \ref{Fig2} of the Kerr black holes for $0<J^{2}<m^{4}$. To facilitate the discussion we will assume $m=1$ and therefore $0<a=J^{2}<1$. The space-time is represented in the Penrose diagram in four sectors. The quadrant on the right corresponds to the metric (\ref{Km}) for the range of coordinates 
\ben
\bigg\{-\infty<t<\infty,\ 1+(1-a^{2})^{\frac{1}{2}}<r<\infty,\ 0\leq \theta\leq \pi,\ 0\leq \varphi<2\pi\bigg\},
\een
and the upper quadrant corresponds to the range of coordinates
\ben
\bigg\{-\infty<t<\infty,\ 1-(1-a^{2})^{\frac{1}{2}}<r<1+(1-a^{2})^{\frac{1}{2}},\ 0\leq \theta\leq \pi,\ 0\leq \varphi<2\pi\bigg\}.
\een
Note that in this case $\partial_{t}$ is space-like while $\partial_{r}$ is time-like. 
One can check directly Theorem \ref{MLemma} for the $\{t=0\}$ maximal slice which has the black hole horizon located at $\{r=1+(1-a^{2})^{\frac{1}{2}}\}$.  The area is 
\ben
A=(1+(1-a^{2})^{\frac{1}{2}}) 8\pi,
\een
which is always greater than $8\pi$ and converges to $8\pi$ as $a\uparrow 1$. 
A more interesting phenomenon occurs when we see Corollary \ref{Cor2} in the light of the maximal slices given by the evolution in the maximal gauge of the initial slice. The maximal foliation penetrates inside the black hole region (the upper quadrant) and approaches in the long-time limit to a maximal slice, ``the last slice", lying entirely inside the black hole region [\footnote{We do not know an explicit expression for the last slice.}] (see Figure \ref{Fig2}). One observes now that the area of the spheres of constant $r$ along the hyper-surface $\{t=1+(1-a^{2})^{\frac{1}{2}}\}$ in the upper quadrant, evolve from $(1+(1-a^{2})^{\frac{1}{2}})8\pi$ monotonically to $(1-(1-a^{2})^{\frac{1}{2}})8\pi$. 
The area of the sphere formed by the intersection of the ``last slice" and $\{t=1+(1-a^{2})^{\frac{1}{2}}\}$ thus approaches to $8\pi$ as $a\uparrow 1$. According to Corollary \ref{Cor2} the ``last slice" must approach the Kerr-throat as $a\uparrow 1$. This can be seen explicitly by studying carefully the metric (\ref{Km}) as $a\uparrow 1$. What occurs is a remarkable phenomenon. As $a\uparrow 1$ the metric (\ref{Km}) in the upper quadrant of the Penrose diagram converges into the following metric
\begin{align}
\label{ms} {\bf g}=&-(1+\cos^{2}\bar{\theta})d\bar{t}^{2}+(1+\cos^{2}\bar{\theta})((\tan \bar{t})\bar{r} d\bar{t} +d\bar{r})^{2}\\
\nonumber &+\bigg[\frac{4\sin^{2}\bar{\theta}}{(1+\cos^{2}\bar{\theta})}\bigg](\bar{r}d\bar{t}-d\bar{\varphi})^{2}+(1+\cos^{2}\bar{\theta})^{2}d\bar{\theta}^{2},
\end{align}
where we have barred the space-time coordinates to emphasize that they are not the same as the coordinates $(t,r,\theta,\varphi)$ which indeed degenerate as $a\uparrow 1$. This metric we call a {\it metric soliton} in the sense that, as seen as a flow $(g,K;N,X)(\bar{t})$ over $S^{2}\times \mathbb{R}$ with coordinates $(\bar{r},\bar{\theta},\bar{\varphi})$, we have
\begin{align*}
& g(\bar{t})=g_{T},\vs\\
& K(\bar{t})=K_{T}+(1+\cos^{2}\bar{\theta})^{\frac{1}{2}}(\tan \bar{t}) d\bar{r}^{2},\vs\\
& N(\bar{t})=(1+\cos^{2}\bar{\theta})^{\frac{1}{2}},\vs\\
& X(\bar{t})=\bar{r}(\tan \bar{t}) \partial_{\bar{r}}-\bar{r}\partial_{\bar{\varphi}},
\end{align*}
where, as it is apparent, the metric $g$ does not evolve and remains equal to the three-metric $g_{T}$ of the Kerr-throat given by (\ref{mKt}). The only slice $\{\bar{t}=const.\}$ with zero mean curvature is $\{\bar{t}=0\}$ and is the limit of the ``last slices" as $a\uparrow 1$. As $\tan (\bar{t}=0)=0$ we conclude from the expression above that the data $(g,K)$ over the slice $\{\bar{t}=0\}$ is exactly the Kerr-throat. Note that $t\in (-\pi/2,\pi/2)$. The metric (\ref{ms}) is globally hyperbolic and partly coincides with the so called {\it near horizon geometry} (see \cite{2009JHEP...09..044A}, Eq. (2.5)) that has been extensively studied in the literature and is not globally hyperbolic. The spacetime described by the metric (\ref{ms}) has the remarkable property that is foliated by marginally trapped spheres saturating (\ref{AJI}). In a certain sense the whole solution is a horizon.  
To obtain the expression  (\ref{ms}) make the change of variables $(t,r,\theta,\varphi)\rightarrow (\bar{t},\bar{r},\bar{\theta},\bar{\varphi})$ 
\ben
\bar{t}=\frac{a(-\Delta)^{\frac{1}{2}}}{r^{2}+a^{2}}t,\quad
\bar{r}=\arcsin \frac{1-r}{(1-a^{2})^{\frac{1}{2}}}\quad
\bar{\theta}=\theta,\quad \text{and}\quad 
\bar{\varphi}=\varphi-\frac{a}{r^{2}+a^{2}}t,
\een
in the expression (\ref{Km}) and take the limit as $a\uparrow 1$.

\vs
We glimpse now on the basic idea behind the proof of Theorem \ref{MLemma}. We suppose by contradiction that there is a data set $(\Sigma;g_{0},K_{0})$ as in the hypothesis of Theorem \ref{MLemma} but possessing an embedded (orientable, compact and boundary-less) surface $S=S_{H}$ saturating (\ref{AJI}). 
Then, as discussed before, the surface $S_{H}$ must be an extreme Kerr-throat sphere. Let $(g(t),K(t))$ be the evolution of the initial data $(g_{0},K_{0})=(g(0),K(0))$ in the maximal gauge and with zero shift (see Section \ref{VEE}). In particular every $(\Sigma;g(t),K(t))$ is a maximal axisymmetric and asymptotically flat vacuum data set. Then, as proved by a simple calculation in Proposition \ref{P1}, unless the lapse function $N(0)$ at the time zero is proportional to $\alpha_{T}$ over $S_{H}$
then the $g(t)$-area of the surface $S_{H}$ strictly decreases in short times. But as $J(S_{H})$ is preserved then the universal inequality (\ref{AJI}) would be violated in short times. Unfortunately this simple argument works as long as  $N(0)$ is not proportional to $\alpha_{T}$ over $S_{H}$. It takes some technical work and a large part of this article to provide a proof of the Theorem \ref{MLemma} in the spirit described but which also contemplates this possibility. The proof of Theorem \ref{Lemma3} is based in the same principle. We will be explaining more along the article.        

The article is organized as follows. In Section \ref{BN} we recall the very basic notions and definitions required to read the article. Section \ref{PMR} contains the proof of Theorems \ref{MLemma} and \ref{Lemma3}. In the Appendix we prove some propositions used in the proofs of the main results whose contents are a bit apart from the mainstream of the article.

\section{Basic notions.}\label{BN}

In this section we introduce the basic notions that we will use during the article. First we recall the Einstein vacuum equations from the dynamical point of view. This point of view is the one appropriate to analyze initial data as we will do in Theorems \ref{MLemma} and \ref{Lemma3}. Then we recall the second variation formula for minimal surfaces on three manifolds and the form that they acquire when the metric is the metric of an initial data for vacuum solutions. After that we introduce areal coordinates on axisymmetric spheres and other coordinates which play a crucial role in certain formulae. Finally we discuss basic but important properties of stable surfaces in axisymmetric three manifolds whose proofs are given in the Appendix. The presentation is somehow general and could be of use in further research.      
\subsection{The vacuum Einstein equations.}\label{VEE} 
Let $\Sigma$ be a fixed smooth manifold 
(say covered by a set of fixed charts $(x^{1},x^{2},x^{3})$). Given a Riemannian metric $g$ we will denote by $Ric$ and $R$ the Ricci and the scalar curvatures of $g$ respectively. Given a symmetric two tensor $K$, we will use the notation $|K|^{2}=K_{ij}K^{ij}$ and $k=tr_{g}K=K_{ij}g^{ij}$. Finally given a vector field $X=X^{i}$, ${\mathcal{L}}_{X}$ will denote the Lie derivative along $X$. 
Let $(g,K;N,X)(t)$ be a smooth flow over $\Sigma$ of: 
\begin{align*}
&\text{Riemannian metrics $g(t)=g_{ij}(t)$ and Symmetric two-tensors $K(t)=K_{ij}(t)$},\vs\\ 
&\text{Lapse functions $N(t)>0$ and Shift vector fields $X(t)=X^{i}(t)$}.
\end{align*}
If for all $t$ (in its range $I$) we have
\begin{align}
\label{DEG} & \dot{g}_{ij}=-2NK_{ij}+({\mathcal{L}}_{X}g)_{ij},\\
\label{DEK} & \dot{K}_{ij}=-\nabla_{i}\nabla_{j} N+N(Ric_{ij}+kK_{ij}-2K_{i}^{\ k} K_{kj})+({\mathcal{L}}_{X} K)_{ij},\\
\label{EC} & R=|K|^{2}-k^{2},\\
\label{MC} & \nabla^{i}K_{ij}=\nabla_{j}k,
\end{align}
then the $3+1$ metric
\ben
{\bf g}=-(N^{2}-X_{i}X^{i})dt^{2}+X_{i}(dt\otimes dx^{i}+dx^{i}\otimes dt)+g_{ij}dx^{i}dx^{j},
\een
is a solution of the Einstein vacuum equations (${\bf Ric}=0$) on $I\times \Sigma$ where $I$ is the interval on which $t$ varies \cite{MR583716}. 
If we let $\mathfrak{n}$ be a unit-normal to the level sets of the time coordinate (which are space-like hyper-surfaces) we have $\partial_{t}=N\mathfrak{n}+X^{i}\partial_{i}$. Moreover $K(t)$ are the second fundamental forms of the level sets of the time coordinate, namely of $\{t\}\times \Sigma$. 
Equations (\ref{DEG})-(\ref{DEK}) are the {\it dynamical equations} and (\ref{EC})-(\ref{MC}) are the {\it energy} and {\it momentum constraints equations} respectively. The evolution is said to be {\it maximal} if $k(t)=0$ for all $t$. In such case the Lapse satisfies, at any time, the {\it Lapse equation}
\be\label{Lapse}
\Delta N=|K|^{2}N.
\ee
Conversely if $\Sigma$ is a space-like hyper-surface (possibly with boundary) on a vacuum space-time and $V$ is a non-zero time-like vector field defined on an open space-time neighborhood of $\Sigma$ then one can obtain a flow $\Sigma(t)$, at least for a short time, by moving $\Sigma$ along $V$. Coordinates charts $(x^{1},x^{2},x^{3})$ are propagated by $V$ to every $\Sigma(t)$ and any two $\Sigma(t)$ and $\Sigma(t')$ are naturally diffeomorphic. In this way one obtains a flow $(g_{ij}(t),K_{ij}(t))$ for the induced three-metrics and second fundamental forms on the {\it fixed} manifold $\Sigma$. Writing $V|_{\Sigma(t)}=N(t)\mathfrak{n}+X^{i}(t)\partial_{i}$, where $\mathfrak{n}$ is a space-time unit normal to $\Sigma(t)$, then one obtains a flow of Lapse functions $N(t)$ and Shift vectors $X(t)=X^{i}\partial_{i}$. The flow $(g(t),K(t);N(t),X(t))$ satisfies (\ref{DEG})-(\ref{MC}).    
{\it It is important to stress that the flow is on a fixed manifold $\Sigma$.}

\subsection{Angular momentum.}\label{ANGULARMOMEN}

Let $S$ be an orientable compact and boundary-less surface embedded in an axisymmetric data set $(\Sigma;g,K)$. Let $\varsigma$ be one of the two unit-normal fields to $S$ in $\Sigma$. Then the Komar angular momentum $J(S)$ (in the direction of $\varsigma$) is 
\be\label{KAM}
J(S):=\frac{1}{8\pi}\int_{S} K(\xi,\varsigma)\, dA.
\ee
In this expression $\xi$ is the rotational Killing field, $K$ is the second fundamental form of the data set and $\varsigma$ is a unit normal to $S$. Of course with the other choice of the normal $\varsigma$ the angular momentum just changes sign. Note that by the axisymmetry and (\ref{MC}) we have $\nabla^{i}(K_{ij}\xi^{j})=0$. This shows that $J(S)$ depends only on the homology class of $S$.  

\subsection{Minimal surfaces and the second variation of area.}
Let $(\Sigma;g,K)$ be a data set. An embedded surface $S$ is minimal if its mean curvature is identically zero. Suppose that $S$ is a compact orientable minimal surface possibly with boundary. 
Let $\varsigma$ be a unit normal vector field to $S$ and let  $\alpha:S\rightarrow \mathbb{R}$ be a smooth function that is zero on $\partial S$ when $\partial S\neq \emptyset$. Then, the second variation of the area, $A''_{\alpha}(S)$, for the deformation of $S$ 
along $\alpha\varsigma$ is the well known \cite{MR2780140}
\be\label{SIM}
A''_{\alpha}(S)=\int_{S}\big[\, |d \alpha|^{2}-(|\Theta|^{2}+Ric(\varsigma,\varsigma))\,\alpha^{2}\big]\, dA,
\ee  
where here $\Theta$ is the second fundamental form of $S$. The mean curvature will be denoted by $\meanc$. The minimal surface $S$ is said to be stable if $A''_{\alpha}(S)\geq 0$ for all $\alpha$. In 
dimension three we have the general identity $2{\mathcal{K}}=(\meanc)^{2}-|\Theta|^{2}+R-2Ric(\varsigma,\varsigma)$, where ${\mathcal{K}}$ is the Gaussian curvature of $S$ and of course $\meanc=0$ when $S$ minimal. From this and the energy constraint we deduce that if $S$ is stable then for any $\alpha$ as described before we have
\be\label{SIMD} 
\int_{S} \big(\, |d\alpha|^{2}+{\mathcal{K}} \alpha^{2}\, \big)\, dA\geq \frac{1}{2} \int_{S} \big(\, |\Theta|^{2} +|K|^{2}-k^{2}\, \big)\, \alpha^{{2}}\, dA.
\ee
\subsection{Areal and polar coordinates for axisymmetric spheres and further coordinates.}\label{APCO}  

Let $(\Sigma,g)$ be an a axisymmetric manifold. Then the group $U(1)$ acts on $(\Sigma,g)$ and the orbits are either circles or fixed points. The set of fixed points are a set of complete one dimensional manifolds (i.e. diffeomorphic to lines or circles) which we call the axes and denote by $\poles$.    

\begin{enumerate}
\item[I.] {\it Areal coordinates on spheres}. Let $S\subset \Sigma$ be an axisymmetric sphere. Then $S$ is foliated by $U(1)$-orbits $\{C\}$. Two of these orbits, that we call the poles, are just points and are denoted by $\north$ (from ``North") and $\south$ (from ``South"). Every orbit which is not a pole divides $S$ into two discs that we denote by $D^{\north}(C)$ and $D^{\south}(C)$ (the first contains $\north$ and the second $\south$). Their areas are denoted by $A^{\north}(C)=\text{Area}(D^{\north}(C))$ and $A^{\south}(C)=\text{Area}(D^{\south}(C))$. Define the polar coordinate $\theta(C)$ at $C$ by
\ben
A^{\north}(C)=\frac{A(S)}{2}\big(1-\cos \theta(C)\big).
\een
To define an azimuthal coordinate $\varphi$ proceed as follows. The gradient of $\theta$ (inside $S$) defines a vector field perpendicular to the orbits and invariant under the $U(1)$-action. The integral curves foliate smoothly $S\setminus \{\north,\south\}$. Denote the foliation by $\{C'\}$. Pick any integral curve to define $\{\varphi=0\}$. Then $\varphi(C')$ is the angle necessary to rotate $\{\varphi=0\}$ to get $C'$. In these coordinates the induced metric $h$ (from $g$) on $S$ looks like
\be\label{MH}
h=\bigg[\narea\bigg]^{2}e^{\displaystyle -\sigma}d\theta^{2}+e^{\displaystyle \sigma}\sin^{2}\theta d\varphi^{2},
\ee
where $\sigma=\sigma(\theta)$. Note that $dA=(A(S)/4\pi)\sin\theta d\theta d\varphi$.

\item[II.] {\it Polar coordinates on discs.} In addition to areal coordinates we will use polar coordinates for discs $D^{\north}$ or $D^{\south}$. Around the respective pole the induced metric $h$ in polar coordinates is
\be\label{POLCO}
h=ds^{2}+\bigg[\frac{\ell(s)}{2\pi}\bigg]^{2}d\varphi^{2}.
\ee 
Of course $\{s=0\}$ is the pole and $\ell(s)$ is the length of the orbit $C(s)$ at a $h$-distance $s$ from the pole. Say $\{s=0\}$ is the north pole. Then we will denote $A(s):=A^{\north}(C(s))$. This notation will be used later (see the proof of Proposition \ref{P4}).  In particular $A(s)=\int_{0}^{s} \ell(\bar{s})d\bar{s}$.

\item[III.] {\it Gaussian-coordinates around surfaces.} Let $S$ be an axisymmetric sphere embedded in $\Sigma$ and provided with areal-coordinates $(\theta,\varphi)$. For $\bar{r}>0$ small enough let 
\ben
{\mathcal T}_{g}(S,\bar{r}):=\{p\in \Sigma, dist_{g}(p,S)\leq \bar{r}\},
\een
be the tubular (closed) neighborhood of $S$ of radius $\bar{r}$. For every $q\in S$ let $\gamma_{q}(r)$ be the $g$-geodesic emanating from $q$, perpendicular to $S$ and parametrized with (signed) arc-length $r\in [-\bar{r},\bar{r}]$. For every $p\in {\mathcal T}_{g}(S,\bar{r})$ let $q(p)$ be the initial point in $S$ such that $\gamma_{q(p)}(r(p))=p$ ($|r(p)|=dist_{g}(p,S)$). The $g$-{\it Gaussian coordinates} $(r,\theta,\varphi)$ on ${\mathcal T}_{g}(S,\bar{r})$ are defined through
\ben
(r,\theta,\varphi)(p)=(r(p),\theta(q(p)),\varphi(q(p))).
\een
We will use Gaussian coordinates often in this article. Observe that $\partial_{\varphi}$ is the rotational Killing field. In Gaussian coordinates the metric $g$ is written as $g=dr^{2}+h$ where $h(\partial_{r},-)=h(-,\partial_{r})=0$ at any point. 

\item[IV.] {\it Space-time coordinates.} The Gaussian coordinates are propagated along the evolution (using $(N,X)$) in such a way that $(r,\theta,\varphi,t)$ ($t$ small) is a space-time coordinate system around $S$. 
We will use the index $A,B,\ldots$ when we use only the two coordinates $(\theta,\varphi)$ while we will use the index $i,j,\ldots$ when use the three coordinates $(r,\theta,\varphi)$. In this sense the components of $h$ are $h_{AB}$ while those of $g$ are $g_{ij}$.

\end{enumerate}

Above we introduced ${\mathcal T}_{g}(S,\bar{r})$. More in general, we will use
\ben
{\mathcal T}_{g}(\Omega,\bar{r})=\{p\in \Sigma/ {\rm dist}_{g}(p,\Omega)\leq \bar{r}\},
\een 
to denote the (closed) tubular neighborhood of a set $\Omega\subset \Sigma$ and radius $\bar{r}$. 

\subsection{Homogeneously regular manifolds.}\label{HRM}

A complete manifold is homogeneously regular if the injectivity radius is uniformly bounded from below (away from zero) and the curvature is uniformly bounded from above. A more quantitative but equivalent definition is the following \cite{MR678484}. 
{\it A complete manifold $\Sigma$ is $\rho_{0}$-homogeneously regular if there is $0<\nu<1/2$ such that at any point $p\in \Sigma$ we have
\begin{enumerate}
\item The exponential map $Exp$, from $B_{T_{p}\Sigma}(0,\rho_{0})$ into $B(p,\rho_{0})$ is a diffeomorphism (where $B_{T_{p}\Sigma}(0,\rho_{0})$ is the ball in the tangent space $T_{p}\Sigma$ of center $p$ and radius $\rho_{0}>0$ and $B(p,\rho_{0})$ is the geodesic ball in $\Sigma$ of center $p$ and radius $\rho_{0}$). 
\item Let $\{(\bar{x}^{1},\bar{x}^{2},\bar{x}^{3})\}$ be cartesian coordinates in $T_{p}\Sigma$. Define as usual coordinates on $B(p,\rho_{0})$ by $(x^{1},x^{2},x^{3})=(\bar{x}^{1},\bar{x}^{2},\bar{x}^{3})\circ Exp^{-1}$. Then, 
\ben
\sup_{B(p,\rho_{0})} \big|\,g_{ij}-\delta_{ij}\,\big|\leq \nu,\
\sup_{B(p,\rho_{0})}\bigg|\,\frac{\partial g_{ij}}{\partial x^{k}}\,\bigg|\leq \frac{\nu}{\rho_{0}},\
\sup_{B(p,\rho_{0})}\bigg|\,\frac{\partial g_{ij}}{\partial x^{k}x^{l}}\,\bigg|\leq \frac{\nu}{\rho_{0}^{2}},
\een
\end{enumerate}
for all $i,j,k,l$ in $\{1,2,3\}$.} 

\vs
Essentially, a homogeneously regular manifold is one for which the metric is controlled in $C^{2}$ (on certain coordinates) on every metric ball of a uniform radius.   
It is easy to see from the definition that if $(\Sigma,g)$ is $\rho_{0}$-homogeneously regular then $(\Sigma,\lambda_{0}^{2}g)$, $\lambda_{0}>0$, is $\lambda_{0}\rho_{0}$-homogeneously regular. In particular, for any $\epsilon>0$ there is $\lambda_{0}(\epsilon,\rho_{0})$ such that for any point $p$, the scaled metric $\lambda_{0}^{2}g$ on $B_{\lambda_{0}^{2}g}(p,1)$ is $\epsilon$-close in $C^{2}$ (in the coordinates described above) to the flat metric on the unit ball of $\mathbb{R}^{3}$. Before $B_{\lambda_{0}^{2}g}(p,1)$ is the ball of center $p$ and radius one with respect to $\lambda_{0}^{2}g$. 

For a manifold with boundary the definition is similar, but we require that there is an extension $(\bar{\Sigma},\bar{g})$ of $(\Sigma,g)$ such that the boundary of $\bar{\Sigma}$ lies at a $\bar{g}$-distance greater than $\rho_{0}$ from $\partial \Sigma$ and at every point $p$ of $\Sigma$ the points $1$ and $2$ hold. Of course every compact manifold, with boundary or not, is $\rho_{0}$-homogeneously regular for some $\rho_{0}$. 

\subsection{Minimal surfaces in axisymmetric spaces.}\label{MSAS} Let $B_{\mathbb{R}^{3}}(o,r)$ be the ball in $\mathbb{R}^{3}$ of radius $r$ and centered at the origin $o=(0,0,0)$. We will think $\mathbb{R}^{3}$ as an axisymmetric space-time where the $U(1)$-action is by rotations around the $z$-axis. Every compact, connected and possibly with boundary axisymmetric surface $S$ embedded inside $\overline{B_{\mathbb{R}^{3}}(o,1)}$ (but with its boundary not necessarily in $\partial B_{\mathbb{R}^{3}}(o,1)$) is either 
\begin{enumerate}
\item[${\mathcal{I}}_{0}$.] An axisymmetric sphere or an axisymmetric torus (zero boundary component),
\item[${\mathcal I}_{1}$.] An axisymmetric disc (one boundary component),
\item[${\mathcal I}_{2}$.] An axisymmetric cylinder, namely diffeomorphic to $S^{1}\times [0,1]$ (two boundary components),
\end{enumerate}
Let ${\mathscr{S}}_{i}$, $i=0,1,2$, be the set of connected axisymmetric surfaces embedded in $B_{\mathbb{R}^{3}}(o,1)$ of type ${\mathcal I}_{i},\ i=0,1,2$ respectively and with boundary, if any, lying in $\partial B_{\mathbb{R}^{3}}(o,1)$. Let $C$ be an orbit in $B_{\mathbb{R}^{3}}(o,1)$ which is not a point in the axis. Let ${\mathscr{I}}^{C}_{1}$ be the set of axisymmetric discs in ${\mathcal I}_{1}$ with boundary $C$. 


The following Proposition is straightforward (indeed the first {\it item} is trivial) from standard properties of minimal surfaces. As it is standard and the proof has few to do with the content of the article we divert the proof until the Appendix. 
\begin{Proposition}\label{PMSF} There is $0<L<1$ such that, 
\begin{enumerate}
\item For every orbit $C\subset B_{\mathbb{R}^{3}}(o,L)$, which is not a point in the axis, there is a unique disc $D$ in 
${\mathscr{S}}^{C}_{1}$ which is minimal. Such unique disc minimizes area among all the surfaces in the family of surfaces 
${\mathscr{S}}^{C}_{1}$ and therefore is stable.
\item There are no stable axisymmetric minimal surfaces in the family of surfaces ${\mathscr{S}}_{0}\cup {\mathscr{S}}_{2}$ and intersecting $B_{\mathbb{R}^{3}}(o,L)$.
\end{enumerate}
\end{Proposition} 
Of course one can scale down the catenoid and then restrict it to $B_{\mathbb{R}^{3}}(o,1)$ to show that there are minimal surfaces in the family ${\mathscr S}_{2}$ reaching as close to the origin $o$ as wished. However such surfaces are not going to be stable when they close enough to the axis.  

It will be necessary to dispose of a version of Proposition \ref{PMSF} but for $C^{2}$-perturbed axisymmetric metrics. More precisely
\begin{Proposition}\label{PMSF2} 
Let $g$ be a smooth axisymmetric metric on the Euclidean ball $B_{\mathbb{R}^{3}}(o,2)\subset \mathbb{R}^{3}$. Assume that the polar coordinate system $(z,\rho,\varphi)$ of $\mathbb{R}^{3}$ is adapted to $g$ in the sense that, in these coordinates $g_{ij}$ is independent on $\varphi$.    
Then there is $\epsilon_{0}>0$ and $0<L<1$ such that if $g$ is $\epsilon_{0}$-close in $C^{2}$ to the flat metric in $\mathbb{R}^{3}$ the following is true. 
\begin{enumerate}
\item For every orbit $C\subset B_{g}(o,L)$, which is not a point in the axis, there is a unique disc $D$ in ${\mathscr{S}}^{C}_{1}$ which is minimal. Such unique disc minimizes area among all the surfaces in the family ${\mathscr{S}}^{C}_{1}$.
\item There are no stable axisymmetric minimal surfaces in the family ${\mathscr{S}}_{0}\cup {\mathscr{S}}_{2}$ and intersecting $B_{g}(o,L)$.
\end{enumerate}
\end{Proposition} 
A corollary which we also prove in the Appendix is the following.
\begin{Corollary}\label{COR3}
Let $(\Sigma,g)$ be a $\rho_{0}$-homogeneously regular, non-necessarily compact axisymmetric three-manifold with smooth, compact and axisymmetric boundary. Assume that the (outward) mean curvature of the boundary is bounded below by $\mu_{0}>0$ and that the norm of the second fundamental form is bounded above by $\mu_{1}$. Let $S$ be an orientable, compact and boundary-less, axisymmetric and stable minimal surface embedded in $\Sigma$. Then, the following four statements hold.
\begin{enumerate}
\item $S\subset (\Sigma\setminus {\mathcal T}_{g}(\partial \Sigma,\epsilon_{1}))$ where $\epsilon_{1}=\epsilon_{1}(\rho_{0},\mu_{0},\mu_{1})$. In other words, $S$ lies at a distance greater than $\epsilon_{1}$ from $\partial \Sigma$. 

\item $A(S)\geq A_{0}(\rho_{0},\mu_{0},\mu_{1})>0$. That is, there is a uniform lower bound for the area of $S$.

\item There is $\epsilon_{2}(\rho_{0},\mu_{0},\mu_{1})<\epsilon_{1}/2$ such that $\partial \big( {\mathcal T}_{g}(\poles, \epsilon_{2})\big)\setminus {\mathcal T}_{g}(\partial \Sigma,\epsilon_{1})$ is smooth and foliated by orbits and for any one of such orbits there is a unique area minimizing disc $D$ with boundary the orbit.    
Moreover, if $S$ intersects ${\mathcal T}_{g}(\poles,\epsilon_{2})$ then it does at two of such discs. Thus, either $S$ intersects the axes $\poles$, in which case it does twice and the surface is a sphere, or it lies at a distance from $\poles$ greater than $\epsilon_{2}$.

\item Let $D$ be any disc as in item 3 and and let $h=ds^{2}+(\ell/2\pi)^{2}(s)\ d\varphi^{2}$ be its two-metric in polar coordinates. Then $s$ ranges in an interval $(0,s_{D}]$ where $0<s_{0}(\rho_{0},\mu_{0},\mu_{1}) <s_{D}< s_{1}(\rho_{0},\mu_{0},\mu_{1})$. Moreover we have $|\ell(s)-2\pi s|\leq c_{0}(\rho_{0}) s^{2}$.
\end{enumerate}
When $\Sigma$ is $\rho_{0}$-homogeneously regular but boundary-less then the estimates in the items 2,3 and 4 depend only on $\rho_{0}$. 
\end{Corollary}

\vs
To finish this section let us mention that orientable, compact and boundary-less stable minimal surfaces in either asymptotically flat axisymmetric manifolds or in compact axisymmetric manifolds with non-empty axisymmetric boundary are necessarily axisymmetric [\footnote{We would like to thank the referee for pointing out that a similar argument appears in {\bf Theorem 8.1} in \cite{Andersson:2007fh}.}]. In particular one can replace the hypothesis ``orientable, compact and boundary-less stable axisymmetric minimal surface" in Corollary \ref{COR3} by ``orientable, compact and boundary-less stable minimal surface". 
We explain this important fact in what follows. Let $\Sigma$ be a three-manifold of one of the two mentioned types and let $S$ be an orientable, compact and boundary-less stable minimal surface embedded in $\Sigma$. Let again $\xi$ denote the axial Killing field and $\varsigma$ a unit normal to $S$. Let $\psi=<\xi,\varsigma>$ be the normal component of $\xi$ on $S$. The surface $S$ will be axisymmetric if we can prove that $\psi$ is zero. Let us show this. If $\psi$ is nowhere zero then the $U(1)$-orbits are transversal to $S$ at any of its points which implies that $\Sigma$ must be diffeomorphic to $S\times S^{1}$. This is impossible because $\Sigma$ is either non-compact or with non-empty boundary. It follows that $\psi$ must be somewhere zero. Suppose that $\psi$ is not identically zero in $S$. As $\xi$ is Killing we deduce that the second variation of the area along $\psi\varsigma$ is zero. From this and because $S$ is stable we deduce that the first eigenvalue of the second variation operator [\footnote{The second variation operator is here that obtained from varying (\ref{SIM}) that is $L \phi = -\Delta \phi - (|\Theta|^{2}+Ric(\varsigma,\varsigma))\phi$.}] is zero and that $\psi$ is its eigenfunction. But the eigenfunction of the first eigenvalue is always nowhere zero and we reach a contradiction.

\section{Proof of the main results.}\label{PMR}
\subsection{The second variation of area in time for extreme Kerr-throat spheres.}
The following is a main technical tool that we will use in the proofs of the main results.
\begin{Proposition}\label{P1}
Let $(\Sigma;g,K)$ be a vacuum, axisymmetric and maximal initial data set and let $S_{H}$ be a stable extreme Kerr-throat sphere embedded in $\Sigma$. Let the initial data $(g,K)$ evolve following the vacuum Einstein equations with smooth lapse $N(t)$ and shift $X(t)$ about which we know only that  $N(0)>0$ and that $X(0)=0$. Then, at time $t$ equal to zero we have 
\ben
\label{FVA} \dot{A}(S_{H})=0,\\
\een
and
\ben
\label{SVA} \ddot{A}(S_{H})=-A''_{N}(S_{H})\leq 0,
\een
where $\dot{A}(S_{H})$ and $\ddot{A}(S_{H})$ are the first and second time derivatives of $A_{g(t)}(S_{H})$ respectively. Moreover equality in the inequality in (\ref{SVA}) holds iff in areal-coordinates we have $N(0)\big |_{S_{H}}=c \alpha_{T}$ for a certain  constant $c$.
\end{Proposition} 
\begin{Note} In a vacuum axisymmetric maximal and asymptotically flat data set, every extreme Kerr-throat sphere is necessarily stable. In Proposition \ref{P1} there is no global assumption on the initial data of any kind and the stability of $S_{H}$ has to be imposed a priori. Nevertheless, as can be seen from the proof, even if $S_{H}$ is not stable we still have $\ddot{A}(S_{H})=-A''(S_{H})$.
\end{Note}
\begin{proof}[\bf Proof.] In the following calculation one can use for instance the space-time coordinate system introduced in Section \ref{APCO}. At any time $t$ (not only zero) we compute
\be\label{CYA1}
\dot{A}(S_{H})=\frac{1}{2}\int_{S_{H}} \dot{h}_{AB}h^{AB}dA=\int_{S_{H}}\big[\, -N K_{AB}h^{AB}+\frac{1}{2}(\nabla_{A}X_{B}+\nabla_{B}X_{A})h^{AB}\big]\, dA,
\ee
where to obtain the second inequality we used (\ref{DEG}). 
Differentiate in time this expression once more and use that $K_{AB}(0)h^{AB}(0)=0$ (because $S_{H}$ is an extreme Kerr-throat sphere of the initial data) and that $X(0)=0$ to obtain at time $t$ equal to zero the expression
\be\label{CALYARENA}
\ddot{A}(S_{H})=\int_{S_{H}}\big[ -N\dot{K}_{AB}h^{AB}+\frac{1}{2}(\nabla_{A}\dot{X}_{B}+\nabla_{B}\dot{X}_{A})h^{AB}\big]\, dA.
\ee
As $S_{H}$ is totally geodesic in $(\Sigma,g(0))$ then (also at time zero) we have
\ben
\int_{S_{H}} \frac{1}{2}(\nabla_{A}\dot{X}_{B}+\nabla_{B}\dot{X}_{A})h^{AB}dA=\int_{S_{H}} \big[ Div_{S_{H}} \Pi(\dot{X})\big]\, dA=0,
\een
where $\Pi(\dot{X})$ is the projection of $\dot{X}$ to the tangent space of $S_{H}$, and $Div_{S_{H}}\Pi(\dot{X})$ is its divergence as a vector field in $(S_{H},h(0))$. Hence the second term in the r.h.s of (\ref{CALYARENA})  is zero. Use now (\ref{DEK}) in (\ref{CALYARENA}) and recall that the initial data is maximal (i.e. $k(0)=0$) to obtain (at time zero)
\be\label{CYA2}
\ddot{A}(S_{H})=\int_{S_{H}}\big[-N\big(-\nabla_{A}\nabla_{B} N +N (Ric_{AB} - 2K_{Ai}K^{i}_{\ B})\big)\, h^{AB}\big]\, dA.
\ee
For the first term on the r.h.s of this expression we have 
\ben
\int_{S_{H}}N(\nabla_{A}\nabla_{B} N)h^{AB} dA=-\int_{S_{H}}|dN|^{2}dA,
\een
because $S_{H}$ is totally geodesic in $(\Sigma,g(0))$. Here $d N$ is the differential of $N$ in $S_{H}$.
On the other hand we have $Ric_{AB}h^{AB}=R-Ric(\varsigma,\varsigma)$ where $\varsigma$ is a unit normal to $S_{H}$ in $\Sigma$. Finally, using the energy constraint and using that at time zero the only nonzero components of $K_{ij}$ are $K_{\varphi r}$ (and $K_{r\varphi}$) then we have $2K_{Ai}K^{i}_{\ B}h^{AB}=|K|^{2}=R$. Putting all together we obtain the following expression for (\ref{CYA2}) at time zero
\ben
\ddot{A}(S_{H})=\int_{S_{H}}\big[-|dN|^{2}+Ric(\varsigma,\varsigma)N^{2}\big]\, dA=-A''_{N}(S_{H}),
\een
where the second inequality is due to the fact the second fundamental form of $S_{H}$ as a surface in $(\Sigma,g(0))$ is zero. This proves the equality in (\ref{SVA}). The inequality instead follows from the fact that because $S_{H}$ is a stable minimal surface in $(\Sigma,g(0))$ then for any $\alpha:S_{H}\rightarrow \mathbb{R}$ we have 
\be\label{PPP}
A_{\alpha\varsigma}''(S_{H})\geq 0.
\ee 
We prove now the last statement of the proposition. As shown in \cite{2011PhRvL.107e1101D} in any extreme Kerr-throat sphere we have $A''_{\alpha_{T}\varsigma}(S_{H})=0$. From this and (\ref{PPP}) we deduce that the first eigenvalue of the stability operator is zero and that $\alpha_{T}$ is an eigenfunction. But because the eigenspace of the first eigenvalue is one dimensional we deduce that $A''_{\alpha\varsigma}(S_{H})=0$ iff $\alpha$ is proportional to $\alpha_{T}$. Therefore $A''_{N}(S_{H})=0$ iff $N(0)|_{S_{H}}=\alpha_{T}$ as wished.
 \end{proof}

\subsection{Proof of Theorem \ref{MLemma}.}
Let us recall the main argument behind the proof of Theorem \ref{MLemma}. Assume by contradiction the existence of $(\Sigma;g_{0},K_{0})$ containing an extreme Kerr-throat sphere $S_{H}$ and to fix ideas suppose that $S_{H}$ is isotopic to a sphere at ``infinity" on one of the asymptotically flat ends. 
%
Suppose that the {\it maximal lapse} $N_{m}$, namely the solution to (\ref{Lapse}) which  is asymptotically one at infinity on $\Sigma$, is not proportional to $\alpha_{T}$ over $S_{H}$. Then, evolving the initial data $(g_{0},K_{0})$ in the maximal gauge with zero shift and using Proposition \ref{P1} and the conservation of angular momentum one obtains, in short time, a sphere on an asymptotically flat, axisymmetric, maximal data violating (\ref{AJI}) which is impossible. 
Unfortunately it could be (a priori) that $N_{m}$ is proportional to $\alpha_{T}$ over $S_{H}$ and therefore the argument, as such, is incomplete. However one can technically modify it, still following the same ground idea, to make it work. 
The modification consists in working on a large but compact region of the initial data enclosed by large convex spheres, where as we will see using Proposition \ref{P2} it is always possible to find a positive solution $N_{0}$ of the Lapse equation not proportional to $\alpha_{T}$ over $S_{H}$. 
Then flow the initial data over such compact region along certain (see later) axisymmetric lapse and shift $(N,X)(t)$ with $N|_{t=0}=N_{0}$ and $X|_{t=0}=0$. Unfortunately the flow will be known to be maximal only at time zero and at later times maximality could fail. This is an important drawback because although by Proposition \ref{P1} we have, for $t$ small, 
\be\label{IF1}
A(S_{H})\leq 8\pi|J(S_{H})|-\frac{A_{N_{0}}''(S_{H})}{4}t^{2},\ \ A_{N_{0}}''(S_{H})>0,
\ee
(where here and below $A(S_{H}):=A_{g(t)}(S_{H})$), the inequality $A(S)\geq 8\pi|J(S)|$ is not known to hold on non-maximal slices, and the original contradiction argument may be inapplicable. Here is where we use that the initial slice is maximal and that, because $N|_{t=0}=N_{0}$ satisfies the Lapse equation and $\dot{k}|_{t=0}=-\Delta N_{0}+|K_{0}|^{2}N_{0}=0$, then the mean curvature $k(t)$ at small times behaves as $k(t)\approx O(t^{2})$. This order of failure of maximality allows us to prove in Proposition \ref{P4}, and for small times, the lower estimate $A(S_{t})\geq 8\pi|J(S_{t})|-\Lambda_{0}t^{4}$ for the area of stable axisymmetric minimal surfaces $S_{t}$.
This is then used in the proof of Theorem \ref{MLemma} to obtain the inequality 
\be\label{IF2}
A(S_{H})\geq 8\pi |J(S_{H})|+O(t^{4}),
\ee
for small times. Inequalities (\ref{IF1})-(\ref{IF2}) show a contradiction in the original spirit. We move then to prove the preliminary Propositions \ref{P2} and \ref{P4}, the proof of Theorem \ref{MLemma} is given afterwards.  

\begin{Proposition}\label{P2} Let $(\Omega;g,K)$ be a vacuum maximal data set where $\Omega$ is a compact manifold with smooth boundary. Suppose that either
\begin{enumerate}[labelindent=\parindent,leftmargin=*,label={\bf A\arabic*.}]
\item There is an extreme Kerr-throat sphere $S_{H}$ dividing $\Omega$ into two connected components $\Omega_{1}$ and $\Omega_{2}$, or, 
\item There are extreme Kerr-throat spheres  $S^{1}_{H}$ and $S^{2}_{H}$, the union of which divides $\Omega$ into two connected components $\Omega_{1}$ and $\Omega_{2}$. 
\end{enumerate}
Then there is an axisymmetric solution $N$ of the Lapse equation which is positive on $\Omega^{\circ}=\Omega\setminus \partial \Omega$ 
and
\begin{enumerate}
\item is not proportional to $\alpha^{1}_{T}$ over $S_{H}$ in case {\bf A1}, or,
\item is not proportional to $\alpha^{1}_{T}$ over at least one of the surfaces $S^{1}_{H}$ or $S^{2}_{H}$ in case {\bf A2}.
\end{enumerate}
where $\alpha^{1}_{T}:=\sqrt{1+\cos^{2}\theta}$ in the areal-coordinates of the respective sphere.
\end{Proposition}
\begin{Note} If $S_{H}$ is an extreme Kerr-throat sphere of area $A_{H}$ then the function $\alpha_{T}$, as defined in the introduction is, $\alpha_{T}=\sqrt{A_{H}/8\pi}\sqrt{1+\cos^{2}\theta}$. Therefore if $N$ is not proportional to $\alpha_{T}^{1}$ over $S_{H}$ then it not proportional either to $\alpha_{T}$. 
\end{Note}

\begin{proof}[\bf Proof.] Denote by $W_{1}$ and $W_{2}$ the set of connected components of $\partial \Omega$ belonging to $\Omega_{1}$ and $\Omega_{2}$ respectively. As $\Omega$ has non-empty boundary by assumption then $W_{1}$ and $W_{2}$ cannot be empty at the same time. We will assume that $W_{2}\neq \emptyset$. We discuss cases {\bf A1} and {\bf A2} separately.  

{\it Case {\bf A1}}. Let $\bar{N}_{a}$ and $\bar{N}_{b}$ be two positive and linearly independent axisymmetric functions over $W_{2}$ and let $N_{a}$ and $N_{b}$ be, respectively, the solutions of the Lapse equation with boundary data 
\begin{align*}
&\Delta N=RN,\\
&N\big |_{W_{2}}=\bar{N}_{a}\ {\rm or}\ \bar{N}_{b},\\ 
&N\big|_{W_{1}}=0,\ \text{if } W_{1}\neq \emptyset.
\end{align*}
The uniqueness of the solutions $N_{a}$, $N_{b}$ and the axisymmetry of the boundary data imply that $N_{a}$ and $N_{b}$ must coincide with their rotational averages and therefore must be rotational symmetric. If either $N_{a}$ or $N_{b}$ are not proportional to $\alpha^{1}_{T}$ over $S$ we are done. On the other hand, if both $N_{a}$ and $N_{b}$ are proportional to $\alpha^{1}_{T}$ over $S_{H}$, we can consider a non-zero linear combination $c_{a}N_{a}+c_{b}N_{b}$ which is equal to zero over $S_{H}$. The combination is also zero over $W_{1}$ if $W_{1}\neq \emptyset$ and as $\partial \Omega_{1}=W_{1}\cup S_{H}$ uniqueness implies that the combination is zero all over $\Omega_{1}$. But because $\bar{N}_{a}$ and $\bar{N}_{b}$ are linearly independent then $c_{a}N_{a}+c_{b}N_{b}$ is not identically zero over $\Omega_{2}$. This contradicts the {\it unique continuation principle for elliptic equations} \cite{MR0092067}. Thus either $N_{a}$ and $N_{b}$ are not proportional to $\alpha^{1}_{T}$ over $S_{H}$.

{\it Case {\bf A2}}. Let $\bar{N}_{a}$, $\bar{N}_{b}$ and $\bar{N}_{c}$ be three positive and linearly independent axisymmetric functions over $W_{2}$. Let $N_{a}$, $N_{b}$ and $N_{c}$ be the axisymmetric solutions to the Lapse equation with boundary data zero over $W_{1}$ if non-empty and 
boundary data $\bar{N}_{a},\ \bar{N}_{b}$ or $\bar{N}_{c}$ over $W_{2}$, respectively. If all of $N_{a},\ N_{b}$ and $N_{c}$ are proportional to $\alpha^{1}_{T}$ over both $S^{1}_{H}$ and $S^{2}_{H}$ then, again, one can easily find a linear combination of the three that is zero over both $S^{1}_{H}$ and $S^{2}_{H}$. A contradiction is then found as in {\it Case {\bf A1}}.
\end{proof}

\vs
Let $(\Sigma;g,k)$ be, as in Theorem \ref{MLemma}, a maximal (vacuum and axisymmetric) data set with possibly finitely many asymptotically flat ends $E_{1},\ldots,E_{n}$, $n\geq 1$. Let ${\mathfrak S}_{1},\ldots,{\mathfrak S}_{n}$ be large axisymmetric and strictly convex spheres on each of the ends $E_{1},\ldots,E_{n}$ respectively. Suppose there is an extreme Kerr-throat sphere $S_{H}$. If $S_{H}$ does not divide $\Sigma$ into two connected components then one can cut $\Sigma$ along $S_{H}$ to get a manifold with two boundary components (say $S^{1}_{H}$ and $S^{2}_{H}$) and glue back smoothly a copy of it (crossing the boundaries). The result is a smooth manifold with $2n$ asymptotically flat ends and having two embedded extreme Kerr-throat spheres $S^{1}_{H}$ and $S^{2}_{H}$ the union of which divides the manifold into two connected components. The conclusion is that if there is an extreme Kerr-throat sphere then either we are in the hypothesis {\bf A1} of Proposition \ref{P2} or that we can construct a data set in the hypothesis {\bf A2}. We will continue as if we were either in the hypothesis {\bf A1} or {\bf A2} therefore. In either case we are denoting by $\Omega$ to the regions enclosed by the large spheres (including the large spheres).          

We note now that there is always a positive solution $N_{0}$ to the Lapse equation on $\Omega$ (including $\partial \Omega$) which is not proportional to $\alpha_{T}$ over $S_{H}$ (in case {\bf A1}) or is not proportional to $\alpha_{T}$ over one of the spheres $S^{1}_{H}$ or $S^{2}_{H}$ (in case {\bf A2}). Indeed let $N_{m}$ be the (positive) solution to the Lapse equation that is asymptotically one over any asymptotically flat end. If $N_{m}$ is proportional to $\alpha_{T}$ over $S_{H}$ (in case {\bf A1}) or is proportional to $\alpha_{T}$ over both, $S^{1}_{H}$ and $S^{2}_{H}$ (in case {\bf A2}) then we can add to it a solution to the Lapse equation as in Proposition \ref{P2} to obtain a positive solution on $\Omega$ with the desired property.   
From now on let $N_{0}$ be the positive axisymmetric solution of (\ref{Lapse}) in $\Omega$, that is not proportional to $\alpha_{T}$ over $S_{H}$ (in case {\bf A1}) or is not proportional to $\alpha_{T}$ over either $S^{1}_{H}$ or $S^{2}_{H}$ (in case {\bf A2}). 

Let ${\mathfrak n}$ be a (one of the two possible) unit-normal fields to $\Omega$ inside the space-time. For any $p\in \Omega$ let $\gamma_{p}(\tau)$ be the (space-time) geodesic emanating from $p$, in the direction of ${\mathfrak n}(p)$, and parametrized with arc length $\tau$. In the domain $\{o/o=\gamma_{p}(\tau),p\in \Omega,0\leq \tau\leq \tau_{0}\}$, $\tau_{0}$ a small constant, we consider a vector field $V$ by
\be\label{VFV}
V(o):=N_{0}(p)\gamma'_{p}(\tau),\ \ {\rm if}\ \gamma_{p}(\tau)=o.
\ee     
Flowing the domain $\Omega$ inside the space-time (generated by the data) by the the vector field $V$, we obtain an evolution flow $(g,K;N,X)(t)$ over the fixed $\Omega$ where $t$ is the parameter associated to the flow by $V$ (see Section \ref{VEE}). Moreover we have $X_{t=0}=0$ and $N_{t=0}=N_{0}$. Note that the evolution induced by $V$ is axisymmetric, namely $(g,K;N,X)$ are axisymmetric. However the evolution does not have to be maximal, that is, we do not necessarily have $k(t)=0$ for all $t$. Despite of this we have
\ben
k\big|_{t=0}=0,\ \ \dot{k}\big|_{t=0}=-\Delta N_{0}+RN_{0}=0.
\een
Thus there is $k_{0}>0$ such that $|k(t)|\leq k_{0}t^{2}$. Therefore the energy constraint implies (at time $t$) $R=|K|^{2}+O(t^{4})$.  The following auxiliary proposition gives a crucial lower bound on the area $A(S_{t})$ of stable minimal surfaces $S_{t}$ in $(\Omega;g(t),K(t))$ (with an upper bound on their areas).

\begin{Proposition}\label{P4} For any $A_{1}>0$ there are $t_{0}>0$ and $\Lambda_{0}>0$ such for any $0<t<t_{0}$ and  (orientable, compact and boundary-less) stable axisymmetric minimal surface $S_{t}$ on $(\Omega;g(t),K(t))$ with $A(S_{t})\leq A_{1}$ we have
\ben
A(S_{t}) \geq 8\pi|J(S_{t})| - \Lambda_{0} t^{4}.
\een
\end{Proposition}
\begin{proof}[\bf Proof] Let $\bar{t}_{0}$, $\rho_{0}$, $\mu_{0}$ and $\mu_{1}$ be such that for any $t$ in $[0,\bar{t}_{0}]$ the manifold $(\Omega,g(t))$ is $\rho_{0}$-homogeneously regular and with strictly convex boundary of (outward) mean curvature greater or equal than $\mu_{0}$ and norm of the second fundamental form bounded by $\mu_{1}$. Below we are going to use the following constants. Let $\epsilon_{2}$ and $A_{0}$ be as in Corollary \ref{COR3}. Let $k_{0}>0$ be (as before) such that for all $t\in [0,\bar{t}_{0}]$, we have $|k(t)|\leq k_{0}t^{2}$. Finally let $\ell_{0}=\sup\{2\pi|\xi(p)|_{g(t)},p\in \Omega,t\in [0,\bar{t}_{0}]\}$. 

For the proof we distinguish two cases, (I) when $S_{t}$ is an axisymmetric torus, and, (II) when $S_{t}$ is an axisymmetric sphere (there are no other possibilities). Take into account that all the calculations below are made on $(\Omega;g(t),K(t))$ in particular that $A(S_{t})$ denotes the $g(t)$-area of $S_{t}$, i.e. $A(S_{t})=A_{g(t)}(S_{t})$.

\vs
(I) {\it $S_{t}$ is an axisymmetric torus}.  Choosing $\alpha=1$ in the stability inequality (\ref{SIMD}) we obtain
\ben
 \int_{S_{t}} |K|^{2} dA\leq \int_{S_{t}} k^{2} dA,
\een
and from (\ref{KAM}) we get 
\ben
(8\pi)^{2}|J(S_{t})|^{2}=|\int_{S_{t}} K(\xi,\varsigma)dA|^{2}\leq \frac{\ell_{0}^{2} A_{1}}{(2\pi)^{2}} \int_{S_{t}} |K|^{2}dA.
\een
From these two inequalities and the bounds mentioned above we obtain
\ben
(8\pi)^{2}|J(S_{t})|^{2}\leq \frac{A^{2}_{1} \ell_{0}^{2}k_{0}^{2}}{(2\pi)^{2}}\ t^{4}. 
\een
Now, if $t_{0}$, ($t_{0}<\bar{t}_{0}$) is small enough we obviously have
\ben
A_{0}\geq  \frac{A_{1} \ell_{0}k_{0}\ t_{0}^{2}}{2\pi}.
\een
Putting all the inequalities together we obtain for any $t\leq t_{0}$ and $\Lambda_{0}>0$
\ben
A(S_{t})\geq A_{0}\geq \frac{A_{1} \ell_{0}k_{0}\ t_{0}^{2}}{2\pi}\geq 8\pi |J(S_{t})| \geq 8\pi |J(S_{t})|-\Lambda_{0} t^{4}.
\een

\vs
(II) {\it $S_{t}$ is an axisymmetric sphere.} Again, take into account that all the calculations below are made on $(\Omega;g(t),K(t))$. Let $(\theta,\varphi)$ be the areal-coordinates of the sphere $S_{t}$ and let $\sigma$ be as in (\ref{MH}). We would like to obtain first uniform upper and lower bounds for the function $\sigma(\theta)$ ($\theta\in [0,2\pi]$) where by uniform we mean independently of $\theta$ and independent also of the stable surface $S_{t}$ ($t_{0}\leq \bar{t}_{0}$, $\bar{t}_{0}$ as before), although dependent on $\rho_{0}, \mu_{0}$ and $\mu_{1}$. Observe that because any axisymmetric sphere has two poles then $S_{t}$ intersects $\poles$ at two points.  By Corollary \ref{COR3} if $S_{t}$ intersects ${\mathcal T}_{g(t)}(\poles,\epsilon_{2})$ it does so in discs $D$ on which, in polar coordinates, we have the bounds 
\begin{gather}
\label{I1} |\ell(s)-2\pi s|\leq c_{0}s^{2},\\
\label{I2} |A(s)-\pi s^{2}|\leq c_{1}s^{3},
\end{gather}
where we are using the notation for $\ell(s), A(s)$ as explained in Section \ref{APCO}. The equations are
valid on $0<s<s_{D}$ with $s_{D}\in [s_{0},s_{1}]$ and the constants $s_{0},s_{1},c_{0},c_{1}$ depend only on $\rho_{0}$, $\mu_{0}$ and $\mu_{1}$. Moreover we have (suppose without loss of generality that $\theta=0$ is the pole of the disc),
\begin{gather}
\label{II1}\ell(s(\theta))=2\pi e^{\scb{.9}{$\sigma(\theta)/2$}}\sin\theta,\\
\label{II2}A(s(\theta))=\frac{A(S_{t})}{2} (1-\cos\theta).
\end{gather}
Combining (\ref{II1}) and (\ref{II2}) we get
\ben\label{COMBINING}
e^{\scb{.9}{$\sigma(\theta(s))$}}=\frac{1}{16\pi^{2}} \bigg[ \frac{A(S_{t})^{2}}{A(S_{t}) - A(s)}\bigg] \bigg[ \frac{\ell^{2}(s)}{A(s)}\bigg].
\een
To estimate the right hand side of this expression we will use
\ben
4\pi \frac{(1-c_{0}s/2\pi)^{2}}{(1+c_{1}s/\pi)} \leq \frac{\ell^{2}(s)}{A(s)}\leq 4\pi \frac{(1+c_{0}s/2\pi)^{2}}{(1-c_{1}s/\pi)},
\een
obtained from the inequalities (\ref{I1})-(\ref{I2}) and $A_{0}\leq A(S_{t})\leq A_{1}$. From them one easily shows that there is $s_{2}(\rho_{0},\mu_{0},\mu_{1})$, with $s_{2}\leq s_{0}$, such that for any $s\in (0,s_{2}]$ 
we have (the coarse) bounds
\be\label{CBOU}
\frac{A_{0}}{8\pi}\leq e^{\scb{.9}{$\sigma$}}\leq \frac{A_{1}}{\pi}. 
\ee
It follows that there is some uniform $\epsilon_{3}(\rho_{0},\mu_{0},\mu_{2})$ with $\epsilon_{3}\leq \epsilon_{2}$ such that for any $S_{t}$, $|\sigma|$ is uniformly bounded on $S_{t}\cap {\mathcal T}_{g(t)}(\poles,\epsilon_{3})$.

Moreover, combining (\ref{I1})-(\ref{II2}) we obtain that there is $\theta_{0}(\rho_{0},\mu_{0},\mu_{2})$ such that for any point $q\in S_{t}$ outside ${\mathcal T}_{g(t)}(\poles,\epsilon_{3})$ we have either $|\theta(q)|\geq \theta_{0}$ or $|\pi-\theta(q)|\geq \theta_{0}$. On the other hand there are uniform upper and lower bounds for the length $\ell(q)$ of the axisymmetric circles (orbits) passing through any point $q\in S_{t}$ outside ${\mathcal T}_{g(t)}(\poles,\epsilon_{3})$. Because of these two facts and the expression 
$e^{\sigma(\theta(q))/2}=\frac{\ell(q)}{2\pi \sin\theta(q)}$ we deduce that there is a uniform bound for $|\sigma(\theta(q))|$ at any point $q\in S_{t}$ outside the tubular neighborhood ${\mathcal T}_{g(t)}(\poles,\epsilon_{3})$. We thus obtain a uniform bound for $|\sigma|$ on any axisymmetric stable minimal sphere $S_{t}$ with $A(S_{t})\leq A_{1}$ as desired. 

Now, following \cite{2011PhRvL.107e1101D} (see also \cite{2011PhRvD..84l1503J}) one can use the stability inequality to deduce [\footnote{Use (15)-(29)-(31)-(32) in \cite{2011PhRvL.107e1101D} and instead of (30) in \cite{2011PhRvL.107e1101D} use the constraint equation $R=|K|-k^{2}$. Finally note that the equation (\ref{KAM}) for the angular momentum is still valid even if the data set is not maximal. Indeed the Komar angular momentum can be intrinsically defined as (5) in \cite{2011PhRvD..84l1503J} which coincides with (\ref{KAM}).}] for axisymmetric spheres $S_{t}$ the following crucial inequality
\be\label{AI1}
A(S)\geq 4\pi e^{\scalebox{1.1}{$\frac{{\mathcal{\overline{M}}}-8}{8}$}},
\ee
where ${\mathcal{\overline{M}}}$ is defined by
\ben
{\mathcal{\overline{M}}}=\frac{1}{2\pi}\int_{S} (\sigma'^{2}+4\sigma +\frac{\omega'^{2}}{\eta^{2}})\sin\theta d\theta d\phi -\frac{1}{2}e^{\scb{.9}{$2c$}}\int_{S} k^{2}e^{\scalebox{1}{$-\sigma$}}\sin\theta d\theta d\phi.
\een
In this expression $e^{c}=A(S)/4\pi$, $\eta=e^{\sigma}\sin^{2}\theta$, the prime in $\sigma'$ and $\omega'$ are their $\theta$-derivative and $\omega=\omega(\theta)$ is a smooth function defined through
\ben
\frac{d\, \omega}{d\, \theta}=\frac{A(S)}{2\pi}K(\xi,\varsigma)\sin\theta.
\een 
A direct computation shows $J(S)=(\omega(\pi)-\omega(0))/8$. On the other hand it is proved in \cite{2011PhRvL.107e1101D}, \cite{2011CQGra..28j5014A} that, if we let
\ben
{\mathcal{M}}=\frac{1}{2\pi}\int_{S} (\sigma'^{2}+4\sigma +\frac{\omega'^{2}}{\eta^{2}})\sin\theta d\theta d\phi,
\een
\n then 
\be\label{AI2}
4\pi e^{\scalebox{1.1}{$\frac{{\mathcal{M}}-8}{8}$}  }\geq 8\pi |J(S)|.
\ee
Using this in (\ref{AI1}) we obtain
\be\label{III}
e^{\big[ \scalebox{1}{$\frac{  \scb{1}{$e^{2c}$}  }{16}\int_{S} k^{2}e^{\scb{.9}{$-\sigma$}}\sin\theta d\theta d\phi$}\big]} A(S)\geq 8\pi |J(S)|.
\ee
We will use this now for $S=S_{t}$ on $(\Omega;g(t),K(t))$. Using the bounds $A(S_{t})\leq A_{1}$, $|k|=|k(t)|\leq k_{0} t^{2}$ and (\ref{CBOU}) we obtain the following bound on the exponent of the l.h.s of the inequality (\ref{III}) 
\be\label{LOLITA}
\left(\frac{A(S_{t})}{16\pi}\right)^{2}\int_{S_{t}} k^{2}e^{\displaystyle -\sigma}\sin\theta d\theta d\phi\leq \frac{A_{1}^{2}k_{0}^{2}}{8A_{0}} t^{4}.
\ee
Let $x=A_{1}^{2}k_{0}^{2}t^{4}/8 A_{0}$. Then if $t_{0}$ is sufficiently small we have $0<x\leq A_{1}^{2}k_{0}^{2}t_{0}^{4}/8 A_{0}\leq 1$ and therefore $e^{x}\leq 1+2x$. Use now (\ref{LOLITA}) in (\ref{AI2}), then $e^{x}\leq 1+2x$ and finally again the bound $A(S_{t})\leq A_{1}$ to  obtain
\ben
A(S_{t})\geq 8\pi|J(S_{t})|- \frac{A_{1}^{3}k_{0}^{2}}{4A_{0}} t^{4}.
\een
The claim follows by defining $\Lambda_{0}=A_{1}^{3}k_{0}^{2}/4A_{0}$.\end{proof}

\vs
We are ready for the proof of Theorem \ref{MLemma}. We recall first the setup of the proof. We assume by contradiction that there is an asymptotically flat data set $(\Sigma;g,K)$ on which there is an extreme Kerr-throat sphere. Then, as was explained, one can always consider a data set $(\Omega;g,K)$ (constructed from the data $(\Sigma;g,K)$) where $(\Omega;g)$ is a compact manifold with strictly mean convex boundary and having an extreme Kerr-throat sphere $S_{H}$ in its interior. Moreover there exists a positive solution $N_{0}$ to the Lapse equation on $\Omega$ which is not proportional to $\alpha_{T}$ over the extreme sphere. The data $(\Omega;g,K)$ is embedded in a space-time and we consider its evolution under the vector field $V$ as in (\ref{VFV}) which gives us an axisymmetric flow $(\Omega;g(t),K(t))$.  

\vs
\begin{proof}[\bf Proof of Theorem \ref{MLemma}] Let $t_{0}$ and $\Lambda_{0}$ be as in Proposition \ref{P4} when one choses $A_{1}=A_{g(0)}(S_{H})$. Let $t$ be a time in $(0,t_{0})$. Below we will work on $(\Omega;g(t),K(t))$. Therefore keep in mind that all the quantities, in particular areas, are found from $(g(t),K(t))$.  

First we observe that the infimum of the areas of all the surfaces isotopic to $S_{H}$ is non-zero. Indeed for any surface $S'$ isotopic to $S_{H}$ we have
\ben
0<|J(S_{H})|=|J(S')|=\frac{1}{8\pi}\bigg|\int_{S'} K(\xi,\varsigma)\, dA\, \bigg|\leq \frac{1}{8\pi}\|K\|_{L^{\infty}}\|\xi\|_{L^{\infty}} A(S').
\een 
Then, following \cite{MR678484} (Theorems 1 and 1')  there is a sequence of surfaces $\{S'_{l}\}$, with each $S'_{l}$ isotopic to $S_{H}$, converging in measure to $n_{1}S_{1}+\ldots+n_{k}S_{k}$, where $\{S_{1},\ldots,S_{k}\}$ is a set of compact and embedded surfaces [\footnote{A very accurate description of the relation between the sequence $\{S'_{l}\}$ and the surfaces $S_{1},\ldots,S_{k}$ is given in Remark (3.27) of \cite{MR678484}.}], [\footnote{Namely for every function $f$ on $\Omega$ we have $
\lim \int_{\bar{S}_{l}}f dA=\sum n_{i} \int_{S_{i}}f dA$.}]. Moreover the infimum of the areas of all the surfaces isotopic to $S_{H}$ is equal to $n_{1}A(S_{1})+\ldots+n_{k}A(S_{k})$. In particular $A(S_{i})\leq A(S_{H})\leq A_{1}$ (this upper bound is needed to apply later Proposition \ref{P4}).
If one of the surfaces, say $S_{i}$, is orientable, then it is stable and therefore axisymmetric (see Section \ref{MSAS}). In such case there are $n_{i}^{+}\geq 0$, $n_{i}^{-}\geq 0$ with $n_{i}=n_{i}^{+}+n_{i}^{-}$, indicating how many times the sequence $\{S'_{l}\}$ ``wraps around" the oriented $S_{i}$ with one orientation and how many with the opposite orientation. Precisely, 
for any two-form $\chi$ supported on a small neighborhood of the oriented surface $S_{i}$, we have
\ben
\lim \int_{S'_{l}} \chi = (n_{i}^{+}-n^{-}_{i}) \int_{S_{i}} \chi.
\een
If on the other hand one of the surfaces, say $S_{i}$, is non-orientable then for any two-form supported in a small neighborhood  $S_{i}$ we have
\be\label{LSL}
\lim \int_{S'_{l}}\chi=0.
\ee
Note that if all the $S_{i}$'s were non-orientable, then using (\ref{LSL}) with $\chi=*K(\xi,-)$, we would get $J(S_{H})=J(S'_{l})=\lim J(S'_{l})=0$ which is not possible. We deduce that at least one of $S_{i}$'s has to be orientable (this fact is not essential). Let us order the surfaces in such a way that $\{S_{1},\ldots,S_{j}\}$, $j\geq 1$ are the orientable (and oriented) and $\{S_{j+1},\ldots,S_{k}\}$ are the non-orientable. We have
\begin{align*}
|J(S_{H})|&=\lim |J(S'_{l})|=| \sum_{i=1}^{i=j}(n^{+}_{i}-n^{-}_{i})J(S_{i})|\leq \sum_{i=1}^{i=j}n_{i}|J(S_{i})|
\leq \sum_{i=1}^{i=j}\frac{n_{i}}{8\pi}A(S_{i})+O(t^{4})\\
&\leq \frac{1}{8\pi}A(S_{H})+O(t^{4})\leq |J(S_{H})|-\frac{A''_{N_{0}}(S_{H})}{16\pi}t^{2}+O(t^{3}),
\end{align*}
where $A''_{N_{0}}(S_{H})>0$ and where to obtain the inequality between the fourth and fifth terms we have used the Proposition \ref{P4} and to obtain the inequality between the sixth and seventh terms we have used Proposition \ref{P1}. 
We obtained thus a contradiction for short times. This finishes the proof of Theorem \ref{MLemma}.\end{proof} 
%


%
\subsection{Proof of Theorem \ref{Lemma3}.}\label{PT2}
In the following sections the reader may benefit from the ``quotient" viewpoint on the geometry of $(\Sigma,g)$, where $(\Sigma,g)$, as in the hypothesis of Theorem \ref{Lemma3}, is axisymmetric and diffeomorphic to $S^{2}\times \mathbb{R}$. 

Recall that the group $U(1)$ acts by isometries and that the set of fixed points consist of two connected and complete one-dimensional manifolds (the axes), and therefore each diffeomorphic to $\mathbb{R}$. The quotient of $\Sigma$ by the action, denoted by $\tilde{\Sigma}$ is diffeomorphic to $[0,1]\times \mathbb{R}$, where ${\mathcal S}:=\{0\}\times \mathbb{R}$ and ${\mathcal N}:=\{1\}\times \mathbb{R}$ are the pair of ``South" and ``North" axis. The set of both axis will be denoted as before by ${\mathcal P}={\mathcal S}\cup {\mathcal N}$ and the topological interior by $\tilde{\Sigma}^{\circ}:=\tilde{\Sigma}\setminus \poles=\tilde{\Sigma}\setminus \partial \tilde{\Sigma}$. We denote the projection by $\Pi$ (in particular $\Pi(\Sigma)=\tilde{\Sigma}$). Denote by $\lambda^{2}$ the square norm of the axisymmetric Killing field and let $\tilde{g}$ be the two dimensional quotient metric on $\tilde{\Sigma}^{\circ}$, namely, if $\tilde{w}=\Pi(w)$ and $\tilde{v}=\Pi(v)$ with $w,v$ tangent vectors at $p\in \Sigma\setminus \poles$, then at $\tilde{p}=\Pi(p)$ we have
\ben
\tilde{g}(\tilde{w},\tilde{v})=g(w,v)-\frac{g(\xi,w)g(\xi,v)}{\lambda^{2}}.
\een      
The metric $\tilde{g}$ extends smoothly to $\tilde{\Sigma}$ \cite{Chrusciel:2007dd}. Every axisymmetric sphere $S$ projects into a one-dimensional manifold diffeomorphic to $[0,1]$ starting and ending $\tilde{g}$-perpendicularly to the axis (see Figure \ref{Fig3}). An axisymmetric sphere is contractible inside $\Sigma$ iff the projection starts and ends in the same axis. Axisymmetric torus project into closed curves inside $\tilde{\Sigma}^{\circ}$ and are therefore contractible in $\Sigma$. Besides spheres and tori, there are no more orientable axisymmetric boundary-less surfaces. Axisymmetric discs project into one-dimensional manifolds diffemorphic to $[0,1]$, starting $\tilde{g}$-perpendicularly at an axis and ending at an interior point. We will use the notation $\beta^{\north,\south}$ for projected axisymmetric spheres starting in ${\mathcal N}$ and ending in ${\mathcal S}$, and $\beta^{\north}(\tilde{p})$ (resp. $\beta^{\south}(\tilde{p})$)  for projected discs starting at $\tilde{p}$ and ending at ${\mathcal N}$ (resp. ${\mathcal S}$). Observe that these curves are embedded. All this is shown in Figure \ref{Fig3}.

Let $\beta(\tilde{s})$ be a curve in $\tilde{\Sigma}$ parametrized with respect to $\tilde{g}$ arc-length, then 
\ben
A(\Pi^{-1}(\beta))=2\pi \int_{s_{0}}^{s_{1}} \lambda(\beta(\tilde{s}))\ d\tilde{s}.
\een 
In other words the area is equal to the length of $\beta$ with respect to the conformal metric $\bar{\tilde{g}}=(2\pi \lambda)^{2}\tilde{g}$. In particular axisymmetric minimal surfaces (which minimize area locally) in $\Sigma$ correspond to $\bar{\tilde{g}}$-geodesics in $\tilde{\Sigma}^{\circ}$. 

\begin{figure}[h]
\centering
\includegraphics[width=10cm,height=5cm]{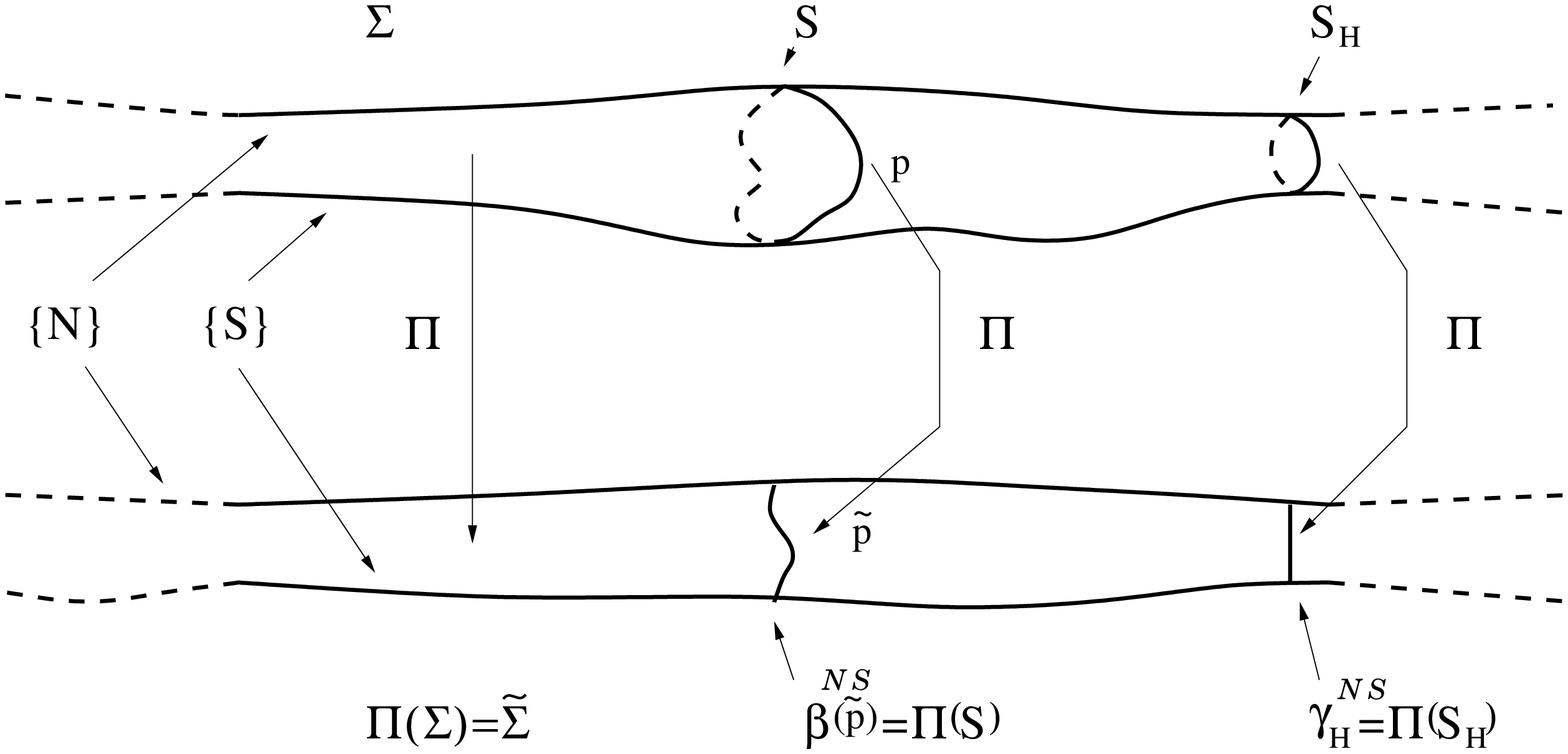}
\caption{}
\label{Fig3}
\end{figure} 

The hypothesis of Theorem \ref{Lemma3} translates into the following two conditions on the quotient manifold:
\begin{enumerate}[labelindent=\parindent,leftmargin=*,label={\bf C\arabic*.}]
\item There is a $\bar{\tilde{g}}$-geodesic $\gamma^{\north,\south}_{H}$ of $\bar{\tilde{g}}$-length $A_{H}:=8\pi|J|$.
\item The $\bar{\tilde{g}}$-length of any curve $\beta^{\north,\south}$ is greater or equal than $A_{H}$.
\end{enumerate}
(we will use $A_{H}$ instead of $8\pi|J|$ from now on). The curve $\gamma^{\north,\south}_{H}$ divides $\tilde{\Sigma}$ into two smooth manifolds $\tilde{\Sigma}_{1}$ and $\tilde{\Sigma}_{2}$ each diffeomorphic to $[0,1]\times \mathbb{R}^{+}_{0}$ ($\mathbb{R}^{+}_{0}=[0,\infty)$). Denote $\Sigma_{i}=\Pi^{-1}(\tilde{\Sigma}_{i})$. At any point $\tilde{p}\in \tilde{\Sigma}_{i}$, with $i$ either $1$ or $2$, define 
\begin{align*}
& \underline{A}_{i}(\tilde{p})=\inf\bigg\{{\rm length}_{\bar{\tilde{g}}}(\beta^{\north,\south}(\tilde{p})),\ \beta^{\north,\south}(\tilde{p})\subset \tilde{\Sigma}_{i}\bigg\},\\
& \underline{A}_{i}^{\north}(\tilde{p})=\inf\bigg\{{\rm length}_{\bar{\tilde{g}}}(\beta^{\north}(\tilde{p})),\ \beta^{\north}(\tilde{p})\subset \tilde{\Sigma}_{i}\bigg\},
\end{align*}
and similarly for $\underline{A}_{i}^{\south}(\tilde{p})$. We also define
\be\label{AU}
\underline{A}(\tilde{p})=\inf \bigg\{{\rm length}_{\bar{\tilde{g}}} (\beta^{\north,\south}(\tilde{p})),\ \beta^{\north,\south}(\tilde{p})\subset \tilde{\Sigma}\bigg\}.
\ee
In the original space, $\underline{A}_{i}(\tilde{p})$ is just the infimum of the areas of the axisymmetric spheres in $\Sigma_{i}$ intersecting ${\mathcal N}$ and ${\mathcal S}$ and containing the orbit $\Pi^{-1}(\tilde{p})$. Similarly $\underline{A}_{i}^{\north}(\tilde{p})$ (resp. $\underline{A}_{i}^{\north}(\tilde{p})$) is the infimum of the areas of the axisymmetric discs in $\Sigma_{i}$ intersecting ${\mathcal N}$ (resp. ${\mathcal S}$) and with boundary $\Pi^{-1}(\tilde{p})$. By {\bf\small{C2}} we have $\underline{A}_{i}(\tilde{p})\geq A_{H}$ for all $\tilde{p}$. The reader can check easily also that $\underline{A}_{i}(\tilde{p})\geq \underline{A}_{i}^{\north}(\tilde{p})+\underline{A}_{i}^{\south}(\tilde{p})$. These quantities are irrelevant outside axisymmetry. 

The use of the quotient picture has considerable advantages but also shortcomings. As a general rule the analysis away from the axis is more conveniently done in the quotient geometry. At the axis however the metric $\bar{\tilde{g}}$ is singular and it is better to stick to the original space. For this reasons we will keep a mixed usage.  

\vs
We state now a basic Proposition concerning area-minimizing sequences of discs that will be required to prove Lemma \ref{LE1} on which the proof of Theorem \ref{Lemma3} relies. To avoid further delays for the proof of Theorem \ref{Lemma3} we postpone the proof of the Proposition until the Appendix. We use the notation $D^{\north}(C)$ and $D^{\south}(C)$ as in Section \ref{APCO} for axisymmetric discs intersecting ${\mathcal N}$ and ${\mathcal S}$ respectively.  
\begin{Proposition}\label{SHGE}
Let $i$ be $1$ or $2$. Let $\Pi(C)=\tilde{p}\in \tilde{\Sigma}^{\circ}_{i}$. Then $\underline{A}_{i}^{\north}(\tilde{p})$ is equal to either
\begin{enumerate}[labelindent=\parindent,leftmargin=*,label={\bf D\arabic*.}]
\item The area of a stable minimal disc $D^{\north}(C)=\Pi^{-1}(\gamma^{\north}(\tilde{p}))$, or,
\item The area of a stable minimal disc $D^{\south}(C)=\Pi^{-1}(\gamma^{\south}(\tilde{p}))$ plus $A_{H}$.
\end{enumerate}
\n 
Moreover there is an area minimizing sequence of discs $D^{\north}_{j}(C)=\Pi^{-1}(\beta^{\north}_{j}(\tilde{p})),\ j\geq 1$ converging (in measure) to either
\begin{enumerate}[labelindent=\parindent,leftmargin=*,label={}]
\item $D^{\north}(C)=\Pi^{-1}(\gamma^{\north}(\tilde{p}))$ in case ${\bf\small{D1}}$ holds, or,
\item $D^{\south}(C)\cup S_{H}=\Pi^{-1}(\gamma^{\south}(\tilde{p}))\cup \Pi^{-1}(\gamma^{\north,\south}_{H})$ in case ${\bf\small{D2}}$ holds.
\end{enumerate}
A similar statement holds for $\underline{A}_{i}^{\south}(\tilde{p})$ by changing $\north\rightarrow \south$ and $\south\rightarrow \north$. 

\end{Proposition}
\begin{Lemma}\label{LE1} Let $i$ be $1$ or $2$. If for all $\tilde{p}\in \tilde{\Sigma}_{i}$ it is $\underline{A}_{i}(\tilde{p})=A_{H}$ then $(\Sigma_{i};g,K)$ is half of the extreme Kerr throat of angular momentum $|J|=A_{H}/8\pi$. 
\end{Lemma}
\begin{proof}[\bf Proof.] Before we start recall that any stable minimal surface $S$ with $A(S)=8\pi|J(S)|$ is an extreme Kerr-throat sphere \cite{2011PhRvL.107e1101D}. 
We prove first that $\Sigma_{i}$ is foliated by extreme Kerr-throat spheres, or, in the quotient space, that 
$\tilde{\Sigma}_{i}$ is foliated by $\bar{\tilde{g}}$-geodesics of length $A_{H}$ starting $\tilde{g}$-perpendicularly to ${\mathcal N}$ and ending $\tilde{g}$-perpendicularly to ${\mathcal S}$. In this first part of the proof we work in the quotient space. Let $\tilde{p}\in \tilde{\Sigma}_{i}^{\circ}$. By hypothesis $\underline{A}_{i}(\tilde{p})=A_{H}$ and recall that $\underline{A}_{i}(\tilde{p})\geq \underline{A}_{i}^{\north}(\tilde{p})+\underline{A}_{i}^{\south}(\tilde{p})$. Therefore, by Proposition \ref{SHGE}, $\underline{A}^{\north}_{i}(\tilde{p})$ and $\underline{A}_{i}^{\south}(\tilde{p})$ are realized by $\bar{\tilde{g}}$-geodesics $\gamma^{\north}(\tilde{p})$ and $\gamma^{\south}(\tilde{p})$ respectively  (case {\bf D1} must hold) and the sum of their $\bar{\tilde{g}}$-lengths is less or equal than $A_{H}$. 

If $\gamma^{\north}(\tilde{p})$ and $\gamma^{\south}(\tilde{p})$ intersect only at $\tilde{p}$ (where they start) and they do not have the same tangent line at $\tilde{p}$ then $\gamma^{\north}(\tilde{p})\cup \gamma^{\south}(\tilde{p})$ could be rounded up at the vertex $\tilde{p}$ to a curve $\beta^{\north,\south}$ with $\bar{\tilde{g}}$-length less than $A_{H}$ violating {\bf C2}. On the other hand $\gamma^{\north}(\tilde{p})$ and $\gamma^{\south}(\tilde{p})$ cannot intersect in a point other than $\tilde{p}$ because in this case one could again construct a curve $\beta^{\north,\south}$ of length less than $A_{H}$ which is not possible [\footnote{$\beta^{\north,\south}$ would be constructed from $\gamma^{\north}(\tilde{p},\tau),\ \tau\in [0,\tau^{\north}]$ and $\gamma^{\south}(\tilde{p},\tau),\ \tau\in [0,\tau^{\south}]$ as follows ($\tau$ here is arc-length and therefore $\tau^{N}+\tau^{\south}\leq A_{H}$). Let $\tau^{\north}_{*}$ be the greatest $\tau$ such that $\gamma^{N}(\tilde{p},\tau)$ is a point also of $\gamma^{\south}$. Suppose that $\gamma^{\north}(\tilde{p},\tau^{N}_{*})=\gamma^{\south}(\tilde{p},\tau^{\south}_{*})$ which defines an $\tau^{\south}_{*}$. Then define $\beta^{\north,\south}(\tau)=\gamma^{\south}(\tilde{p},\tau)$ for $\tau\in [0,\tau^{\south}_{*}]$ and $\beta^{\north,\south}(\tau)=\gamma^{\north}(\tilde{p},\tau-\tau^{\south}_{*}+\tau^{\north}_{*})$ for $\tau\in [\tau^{\south}_{*},\tau^{\north}-\tau^{\north}_{*}+\tau^{\south}_{*}]$. This curve has length less than $A_{H}$ but is not $C^{1}$ at $\beta^{\north,\south}(\tau^{\south}_{*})$. Then round it off at this point to have and embedded curve $\beta^{\north,\south}_{1}$ of length less than $A_{H}$.}]. 
Thus $\gamma^{\north,\south}(\tilde{p}):=\gamma^{\north}(\tilde{p})\cup \gamma^{\south}(\tilde{p})$ is a $\bar{\tilde{g}}$-geodesic of length $A_{H}$ or, the same, $\Pi^{-1}(\gamma^{\north}(\tilde{p})\cup \gamma^{\south}(\tilde{p}))$ is an extreme Kerr-throat sphere. 
We claim that for any $\tilde{p}_{1}\neq \tilde{p}_{2}$ in $\tilde{\Sigma}_{i}^{\circ}$ the geodesics $\gamma^{\north,\south}(\tilde{p}_{1})$ and $\gamma^{\north,\south}(\tilde{p}_{2})$ must be either equal or disjoint. Indeed if they are not disjoint then when they intersect they would have to do transversely and one could again easily construct a curve $\beta^{\north,\south}(\tilde{p})$ of $\bar{\tilde{g}}$-length less than $A_{H}$ violating ${\bf C2}$. The set $\{\gamma^{\north,\south}(\tilde{p}),\tilde{p}\in \Sigma_{i}^{\circ}\}$ is thus the desired foliation of $\tilde{\Sigma}_{i}$ by $\bar{\tilde{g}}$-geodesics of length $A_{H}$ and $\{\Pi^{-1}(\gamma^{\north,\south}(\tilde{p})),\tilde{p}\in \Sigma_{i}^{\circ}\}$ is the desired foliation of $\Sigma_{i}$ by extreme Kerr-throat spheres.

What we have so far is a foliation by extreme Kerr-throat spheres and from this information we want to deduce that there are coordinates $(r,\theta,\varphi)$ on $\Sigma_{i}$ on which the metric $g$ has exactly the expression (\ref{mKt}). We work now in the original manifold $\Sigma_{i}$ (not in the quotient). Let $r$ be a smooth function, constant along the leaves of the foliation and with non-zero gradient everywhere. One can take for instance the function $r$ which at a point $p$ is equal to the volume enclosed by the leaf passing through $p$ and $\partial \Sigma_{i}$. The flow induced by the vector field
\ben
Y:=\frac{\nabla^{i} r}{|\nabla r|^{2}},
\een
takes leaves into leaves because $dr(Y)=1$ (and indeed orbits into orbits because, can be seen, $Y$ is $U(1)$-invariant). Let $(\theta,\varphi)$ be the areal coordinates on $S_{H}$. Extend them to all $\Sigma_{i}$ by Lie dragging, namely define them by imposing $Y(\theta)=0,\ Y(\varphi)=0$. In this way $(r,\theta,\varphi)$ are coordinates on $\Sigma_{i}$. Let $h_{AB}(\bar{r})$ be the metric components of the two-metric induced on the leaf $\{r=\bar{r}\}$ in the coordinates $(\theta,\varphi)$. Then we have
$\partial_{r} h_{AB}=0$,
%
because every leaf is totally geodesic.  Thus $h_{AB}(r)=h_{AB}(0)$ which is the metric of $S_{H}$ in areal coordinates, namely (from (\ref{mKt}))
\ben
h(0)=\bigg(\frac{4|J|\sin^{2}\theta}{1+\cos^{2}\theta}\bigg)\, d\varphi^{2}+|J|(1+\cos^{2}\theta)\, d\theta^{2}. 
\een
Because of  $\partial_{r} h_{AB}=0$ the first and second variation of the area of the leaves along $Y$ is zero. We deduce that at every leaf we must have $\partial_{r}=Y= c(r) \alpha_{T} \varsigma$ where $\varsigma$ is a $g$-unit normal field to the leaf. It follows that one can redefine $r$ to have $\partial_{r}=\alpha_{T} \varsigma$ over every leaf. As $<\partial_{r},\partial_{\theta}>_{g}=<\partial_{r},\partial_{\varphi}>_{g}=0$, the metric $g$ takes in these coordinates the form 
\ben
g=\alpha_{T}^{2}dr^{2}+h_{AB}(0),
\een
which is (\ref{mKt}). That the second fundamental form takes the form (\ref{KKt}) is a direct consequence of the fact that every leaf of the foliation is an extreme Kerr-throat sphere of the same area. \end{proof}

We are ready to prove Theorem \ref{Lemma3}.

\begin{proof}[\bf Proof of Theorem \ref{Lemma3}] To start, let $\Sigma_{1}$ and $\Sigma_{2}$ be the closures of the two connected components of $\Sigma\setminus S_{H}$. For $i=1,2,$ let $\bar{A}_{i}=\sup\{\underline{A}_{i}(\tilde{p}),\tilde{p}\in \Sigma_{i}\}$. Now, if for $i=1,2,$ $\bar{A}_{i}=A_{H}$ then by Lemma \ref{LE1} the data has to be the extreme Kerr throat and we are done. 
Assume then that one of the $\bar{A}_{i}$'s is greater than $A_{H}$. If say $\bar{A}_{1}=A_{H}$ (but $\bar{A}_{2}>A_{H}$) then again by Lemma \ref{LE1} the data over $\Sigma_{1}$ is half of the extreme Kerr-throat data. In this case one can easily make a doubling of the data on $\Sigma_{2}$ and construct a new data on a manifold $\Sigma'$ also diffeomorphic to $S^{2}\times \mathbb{R}$ and having an extreme Kerr-sphere $S'_{H}$ dividing $\Sigma'$ in sectors 
$\Sigma'_{1}$ and $\Sigma'_{2}$ with $\bar{A}'_{1}>A_{H}$ and $\bar{A}'_{2}>A_{H}$. We can then assume without loss of generality that $\bar{A}_{i}>A_{H}$ for $i=1,2$. We will see that this leads to a contradiction.

We observe now that the set of points ${\mathcal E}=\{\tilde{p}\in \tilde{\Sigma}^{\circ}, \underline{A}(\tilde{p})=A_{H}\}$ ($\underline{A}(\tilde{p})$ as in (\ref{AU})) is (i) a closed set, and (ii) a union of projected extreme Kerr-throat spheres.
That ${\mathcal E}$ is closed follows from the fact that $\underline{A}(\tilde{p})$ is continuous with respect to $\tilde{p}$ [\footnote{This is a consequence of the fact that for any two points $\tilde{p}$ and $\tilde{q}$ we have $|\underline{A}(\tilde{p})-\underline{A}(\tilde{q})|\leq 2{\rm dist}_{\bar{\tilde{g}}}(\tilde{p},\tilde{q})$. The reader can check this by proving first that for any $\beta^{\north,\south}(\tilde{p})$ we have, ${\rm length}_{\bar{\tilde{g}}}(\beta^{\north,\south}(\tilde{p}))+2{\rm dist}_{\bar{\tilde{g}}}(\{\beta^{\north,\south}(\tilde{p})\},\tilde{q})\geq \underline{A}(\tilde{q})$ and therefore that ${\rm length}_{\bar{\tilde{g}}}(\beta^{\north,\south}(\tilde{p}))+2{\rm dist}_{\bar{\tilde{g}}}(\tilde{p},\tilde{q})\geq \underline{A}(\tilde{q})$ because ${\rm dist}_{\bar{\tilde{g}}}(\tilde{p},\tilde{q})\geq {\rm dist}_{\bar{\tilde{g}}}(\{\beta^{\north,\south}(\tilde{p})\},\tilde{q})$. Here ${\rm dist}_{\bar{\tilde{g}}}(\{\beta^{\north,\south}(\tilde{p})\},\tilde{q})$ is the $\bar{\tilde{g}}$-distance from the point $\tilde{q}$ to the set curve $\beta^{\north,\south}(\tilde{p})$ (as a set). Taking the infimum among all $\beta^{\north,\south}(\tilde{p})$ we deduce $\underline{A}(\tilde{p})+2{\rm dist}_{\bar{\tilde{g}}}(\tilde{p},\tilde{q})\geq \underline{A}(\tilde{q})$.}]. That ${\mathcal E}$ is a union of projected extreme Kerr-throat spheres follows from the fact, shown inside the proof of Lemma \ref{LE1}, that if $\underline{A}(\tilde{p})=A_{H}$ then $\tilde{p}$ lies in a projected extreme Kerr-throat sphere. 

Now, for $i=1,2$ let $S_{i}$ be a sphere with poles $\north_{i}\in {\mathcal N}$ and $\south_{i}\in {\mathcal S}$ and embedded in $\Sigma_{i}\setminus \Pi^{-1}({\mathcal E})$. Let also $\delta>0$ be such that (i) ${\mathcal T}_{g}(S_{i},4\delta)\subset (\Sigma\setminus \Pi^{-1}({\mathcal E}))$, and (ii) one can construct $g$-Gaussian coordinates inside ${\mathcal T}_{g}(S_{i},4\delta)$.
Let $\Omega$ be the region enclosed by $S_{1}$ and $S_{2}$ including the spheres themselves and observe that $S_{H}\subset \Omega^{\circ}$.
Let $N_{0}$ be a positive solution of the Lapse equation in the region $\Omega_{4\delta}:=\Omega\cup {\mathcal T}_{g}(S_{1},4\delta)\cup {\mathcal T}_{g}(S_{2},4\delta)$ which is not proportional to $\alpha_{T}$ over $S_{H}$ as is provided by Proposition \ref{P2}. Transport now the region $\Omega_{3\delta}:=\Omega\cup {\mathcal T}_{g}(S_{1},3\delta)\cup {\mathcal T}_{g}(S_{2},3\delta)$
inside the space-time following the vector field $V$ as we did in (\ref{VFV}). Then, as was explained in the last paragraph of Section \ref{VEE}, by transporting $\Omega_{3\delta}$ we induce a flow $(g(t),K(t);N(t),X(t))$ over $\Omega_{3\delta}$ for $t\in [0,t_{0}]$ and with $t_{0}$ small. We will use this flow in the argumentation below.
We consider now a smooth path of spheres $S_{i}(t)$, $t\in [0,t_{0}]$ and $i=1,2$ and coinciding at time zero with the $S_{i}$'s introduced before, namely, $S_{i}(0)=S_{i}$, for $i=1,2$. Denote by $\north_{i}(t)$ and $\south _{i}(t)$ the poles of $S_{i}(t)$ and define $\Omega(t)$ as the region enclosed by $S_{1}(t)$ and $S_{2}(t)$. Chose $t_{0}$ smaller if necessary such that for any $t\in [0,t_{0}]$ there are $g(t)$-Gaussian coordinates in ${\mathcal T}_{g(t)}(S_{i}(t),2\delta)$.  
Now, in $\Omega_{2\delta}(t)=\Omega(t)\cup {\mathcal T}_{g(t)}(S_{i}(t),2\delta)\cup {\mathcal T}_{g(t)}(S_{2}(t),2\delta)$ we will consider a flow of axisymmetric metrics $g^{*}(t)$ enjoying the following three properties for every $t\in [0,t_{0}]$,
\begin{enumerate}
\item $g^{*}(t)=g(t)$ on $\Omega_{\delta}(t)=\Omega(t)\cup {\mathcal T}_{g(t)}(S_{i}(t),\delta)\cup {\mathcal T}_{g(t)}(S_{2}(t),\delta)$,
\item $g^{*}(t)\geq g(t)$ on  $\Omega_{2\delta}(t)$,
\item The $g^{*}$-mean curvature of the boundary of the region $\Omega_{2\delta}(t)$ is strictly positive in the outward direction, namely the boundary is strictly mean convex. 
\end{enumerate}
The flow of metrics $g^{*}(t)$ can be explicitly given for instance as follows. On everyone of the two connected components of $\Omega_{2\delta}(t)\setminus \Omega(t)$ we write the metric $g(t)$ in $g(t)$-Gaussian coordinates as
$g(t)=dr^{2}+h_{i}(t)$, $i=1,2$. Then make
\ben
g^{*}(t):=
\left\{ 
\begin{array}{lll}
g(t) & {\rm on} & \Omega_{\delta}(t),\\
dr^{2}+f^{2}(r)h_{i}(t), & {\rm on} & \Omega_{2\delta}(t)\setminus \Omega_{\delta}(t),\ (\delta<r<2\delta),
\end{array}
\right.
\een
where $f(r)$ is the real and time independent function
\ben
f(r)=1+e^{\ \scb{1.1}{$\big[\frac{1}{2\delta+\epsilon-r}-\frac{1}{r-\delta}\big]$}},
\een
for $\epsilon>0$ small enough making the boundary strictly convex. 
\begin{figure}[h]
\centering
\includegraphics[width=13cm,height=5cm]{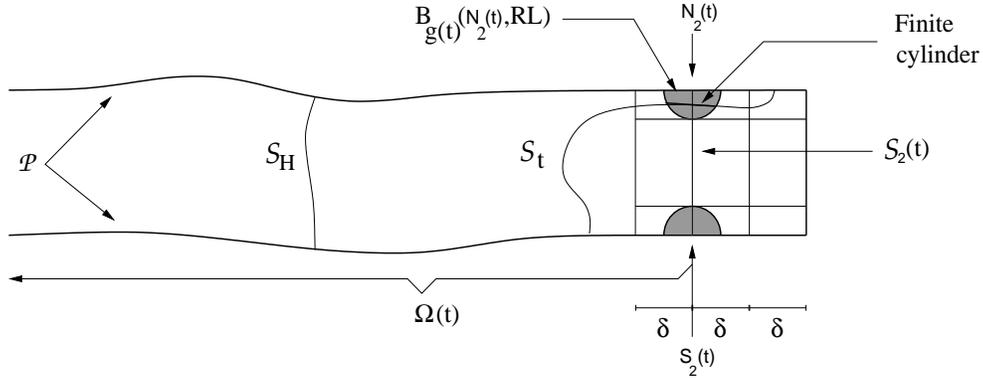}
\caption{Representation of the construction in the argument by contradiction in the proof of Theorem \ref{Lemma3}.}
\label{FFIF}
\end{figure} 
With this setup at hand we move to obtain the main contradiction. By Proposition \ref{P1} and our choice of $N_{0}$ we have $\ddot{A}(S_{H})<0$ at time zero. Because of this we can chose $t_{0}$ smaller if necessary to have $A_{g(t)}(S_{H})<A_{H}$ for any $0<t<t_{0}$. By \cite{MR678484} (Theorems 1 and 1') there is, for every $0<t\leq t_{0}$, a $g^{*}(t)$-stable, axisymmetric and area-minimizing sphere $S_{t}$ inside $\Omega_{2\delta}(t)$ of $g^{*}(t)$-area less than $A_{H}$. We claim that, making $t_{0}$ smaller if necessary, the spheres $S_{t}$ must lie inside $\Omega_{\delta}(t)$ which is a region where by construction the metric $g^{*}(t)$ is equal to $g(t)$. This would prove the $g(t)$-stability of the $S_{t}$ that will be useful later. We show the claim now.

Let $0<R<\delta/2$ be small enough that for every $t\in [0,t_{0}]$ the Riemannian balls $(B_{g(t)/R^{2}}(\north_{i}(t),2),g(t)/R^{2})$ and $(B_{g(t)/R^{2}}(\south_{i}(t),2),g(t)/R^{2})$ are $\epsilon_{0}$-close [\footnote{To be precise in the coordinates $(x^{1},x^{2},x^{3})$ as in Section \ref{HRM}.}] in $C^{2}$ to the flat metric in $B_{\mathbb{R}^{3}}(o,2)$ where $\epsilon_{0}$ is a constant as in Proposition \ref{PMSF2}. Note that because $R<\delta/2$ then the balls are included in ${\mathcal T}_{g(t)}(S_{1}(t),\delta)\cup {\mathcal T}_{g(t)}(S_{2}(t),\delta)$ and that of course $B_{g(t)/R^{2}}(\north_{i}(t),2)=B_{g(t)}(\north_{i}(t),2R)$ and $B_{g(t)/R^{2}}(\south_{i}(t),2)=B_{g(t)}(\south_{i}(t),2R)$. In what follows we let $0<L<1$ be as in Proposition \ref{PMSF2}. We also make $\Omega_{-\delta}(t)=\Omega(t)\setminus ({\mathcal T}_{g(t)}(S_{1}(t),\delta)\cup {\mathcal T}_{g(t)}(S_{2}(t),\delta))$ 

Now, by construction, the closure of the region $\Omega_{2\delta}(0)\setminus (\Omega_{-\delta}(0) \cup {\mathcal T}_{g(0)}(\poles,RL))$ is included in $\Sigma\setminus \Pi^{-1}({\mathcal E})$. Due to this there is $\Gamma>0$ such that for every sphere $S$ with poles in ${\mathcal N}$ and ${\mathcal S}$ and intersecting $\Omega_{2\delta}(0)\setminus (\Omega_{-\delta}(0) \cup {\mathcal T}_{g(0)}(\poles,RL))$ has $g(0)$-area greater or equal than $A_{H}+\Gamma$. By continuity, and making $t_{0}$ smaller if necessary, we can assume that for every $t\in [0,t_{0}]$ every sphere $S$ (with poles in ${\mathcal N}$ and ${\mathcal S}$) intersecting $\Omega_{2\delta}(t)\setminus (\Omega_{-\delta}(t)\cup {\mathcal T}_{g(t)}(\poles,RL))$ has $g^{*}(t)$-area greater or equal than $A_{H}+\Gamma/2$. But the area-minimizing spheres $S_{t}$ have $g^{*}(t)$-area less than $A_{H}$ and therefore if the $S_{t}$ do not lie entirely in $\Omega_{\delta}(t)$ then they must necessarily intersect either the ball
$(B_{g(t)}(\north_{i}(t),RL),g(t))$ or the ball $(B_{g(t)}(\south_{i}(t),RL),g(t))$ and at least one of the connected components of the intersection must be a cylinder with two boundary components (see Figure \ref{FFIF}). This violates Proposition \ref{PMSF2}. Thus $S_{t}\subset \Omega_{\delta}(t)$. 
Now that we have proved that $S_{t}\subset \Omega_{\delta}(t)$ and therefore the $g(t)$-stability of the $S_{t}$ we can proceed in the same way as in Proposition \ref{P4} to show that there is $\Lambda_{0}>0$ such that for any $t\in (0,t_{0}]$ we have $A(S_{t})=A_{g(t)}(S_{t})\geq 8\pi|J|-\Lambda_{0}t^{4}$. On the other hand
$A(S_{t})=A_{g(t)}(S_{t})\leq A_{g(t)}(S_{H})= A_{H} + \ddot{A}(S_{H})t^{2}/2 +O(t^{3})$. This shows a contradiction because $A_{H}=8\pi |J|$ and $\ddot{A}(S_{H})<0$.\end{proof}

\section{Appendix.}

\begin{proof}[\bf Proof of Proposition \ref{PMSF}] We prove first {\it item 1} which is true for any value of $L$ chosen between $(0,1)$ (indeed we just reproduce here the classical proof). Let $D$ be a minimal disc with boundary $C$. Say $z|_{C}=z_{C}$. Then the function $z-z_{C}$ is harmonic on $D$ and $(z-z_{C})|_{C}=0$. It follows that $z=z_{C}$ all over $D$ and therefore that $D$ is the disc enclosed by $C$ in the plane $\{z=z_{C}\}$. Another proof, best suited for extensions, can be obtained along the following lines (we just provide the sketch). 
Let $D(C)$ be the disc enclosed by $C$ in the plane $\{z=z_{C}\}$. Let $o(C)=(0,0,z_{C})$ be its center and $R(C)$ its radius. Let $B^{+}=B_{\mathbb{R}^{3}}(o(C),R(C))\cap \{z>z_{C}\}$ and $B^{-}=B_{\mathbb{R}^{3}}(o(C),R(C))\cap \{z<z_{C}\}$. Now, the disc $D(C)$ is minimal and, by a direct inspection of the stability operator, also strictly stable. Using the strict stability one can construct a smooth foliation of $B^{+}$ by discs $\{D'(C)\}$, each with boundary $C$, and strictly convex (in the direction of increasing $z$) and similarly for $B^{-}$. In addition one can construct a foliation of $B_{\mathbb{R}^{3}}(o,1)\setminus B_{\mathbb{R}^{3}}(o(C),R(C))$ by round spheres $\{S'\}$. We have thus a foliation of $B_{\mathbb{R}^{3}}(o,1)\setminus D(C)$ by strictly convex surfaces, acting as barriers, and preventing the existence of any other minimal disc with boundary $C$ inside $B_{\mathbb{R}^{3}}(o,1)$.   

We prove now {\it item 2}. First we prove that there are no stable surfaces in the class ${\mathscr S}_{0}$. This again is true for any value of $L$ chosen in $(0,1)$. This can be proved as in {\it item 1} by showing that if there is one then the function $z$ has to be constant on it. Another proof, more independent of the flatness of the ambient space $\mathbb{R}^{3}$ and therefore best suited for extensions is the following. Assume again that there is one such surface. Then note that there is a cylinder $\{\rho=\rho_{0}\}$ where $\rho^{2}=x^{2}+y^{2}$, enclosing the surface and tangent to it at least in one orbit (a circle). Such cylinder has strictly mean convex boundary which implies that at the points of tangency the minimal surface must have positive mean curvature (in the outgoing direction from the axis) which is absurd. Note that instead of cylinders one could have used spheres to reach a similar conclusion. 

We prove now that for some $L$ appropriately chosen there are no stable surfaces in the class ${\mathscr S}_{2}$. This requires a bit more effort. Recall that the class ${\mathscr S}_{2}$ consist of axisymmetric cylinders with boundary in $B_{\mathbb{R}^{3}}(o,1)$. From now on we let $S$ be an stable axisymmetric cylinder with boundary in $\partial B_{\mathbb{R}^{3}}(o,1)$. The reader should keep that in mind because it will not be repeated. If $S\cap B_{\mathbb{R}^{3}}(o,R)\neq \emptyset$ for some $R>0$ and $q\in S\cap B_{\mathbb{R}^{3}}(o,R)$ then we will denote by $\big[S\cap B_{\mathbb{R}^{3}}(o,R)\big]^{c.c.}_{q}$ to the connected component of $S\cap B_{\mathbb{R}^{3}}(o,R)$ containing $q$.

In the following we will use two standard results in minimal surfaces that we take form \cite{MR2780140}. We refer the reader to this reference for full details. We first observe that there is a universal constant $c>0$ such that for any $S$ intersecting $B_{\mathbb{R}^{3}}(o,1/8)$ and $q$ in $S\cap B_{\mathbb{R}^{3}}(o,1/8)$ we have
\be\label{SFFE1}
\sup\bigg\{|\Theta|^{2}(p),\ p\in \big[S\cap B_{\mathbb{R}^{3}}(o,1/4)\big]^{c.c.}_{q}\bigg\} \leq 4c.
\ee
This is the result of using Corollary 2.11 in page 79 of \cite{MR2780140} with $r_{0}=7/8$, $\sigma=1/2$ and observing in there that $B_{\mathbb{R}^{3}}(o,1/4)\subset B_{\mathbb{R}^{3}}(q,r_{0}-\sigma)$. This is an important estimate that will be used crucially below.
Let now $L_{0}=\min \{1/(4\sqrt{64c}),1/4\}$ and observe that $4c\leq 1/(16(16L_{0}^{2}))$ and that $L_{0}/2\leq 1/8$. From this and (\ref{SFFE1}) we obtain that for any $S$ intersecting $B_{\mathbb{R}^{3}}(o,L_{0}/2)$ and for any $q$ in $S\cap B_{\mathbb{R}^{3}}(o,L_{0}/2)$ we have
\be\label{33}
16L_{0}^{2}\sup \bigg\{|\Theta|^{2}(p),\ p\in \big[S\cap B_{\mathbb{R}^{3}}(o,L_{0})\big]^{c.c.}_{q}\bigg\}\leq \frac{1}{16}.
\ee
Let $L$ be any number in $(0,L_{0}/8)$. We will see at the end of the argumentation below that if $L<L_{0}/20$ then $S\cap B_{\mathbb{R}^{3}}(o,L)=\emptyset$. At the moment just assume that $0<L<L_{0}/8$. 
We use now the estimate (\ref{33}) in conjunction with Lemma 2.4 in page 74 of \cite{MR2780140} (used with $s: = L_{0}/4$, $\Sigma:=\big[S\cap B_{\mathbb{R}^{3}}(o,L_{0})\big]^{c.c.}_{q}$ and $x:=q$) to conclude that for any surface $S$ intersecting $B_{\mathbb{R}^{3}}(o,L)$ and for any $q$ in $S\cap B_{\mathbb{R}^{3}}(o,L)$ the following two facts hold.
\begin{enumerate}
\item $B_{S}(q,L_{0}/2)$ is a graph of a function $u$ on a domain of $T_{q}S\subset \mathbb{R}^{3}$, where $B_{S}(q,L_{0}/2)$ is the intrinsic ball inside $S$ (with the induced metric) of center $q$ and radius $L_{0}/2$. Moreover $|\nabla u|\leq 1$. 
\item The connected component of $S\cap B_{\mathbb{R}^{3}}(o,L_{0}/4)$ containing $q$, namely $\big[S\cap B_{\mathbb{R}^{3}}(o,L_{0}/4)\big]^{c.c.}_{q}$, lies inside $B_{S}(q,L_{0}/2)$ and therefore is a graph by the item before.
\end{enumerate}
For any surface $S$ such that  $S\cap B_{\mathbb{R}^{3}}(o,L)\neq \emptyset$ denote to simplify notation ${\mathscr C}=[S\cap B_{\mathbb{R}^{3}}(o,L_{0}/4)]^{c.c}_{q}$. Thus ${\mathscr C}$ is a cylinder whose boundary consists of two orbits, $C_{1}$ and $C_{2}$, in $\partial B_{\mathbb{R}^{3}}(o,L_{0}/4)$. Denote by $\Pi_{T_{q}S}$ the projection into the plane $T_{q}S\subset \mathbb{R}^{3}$. Then $\Pi_{T_{q}S}({\mathscr C})$ is an annulus with boundary components $\Pi_{T_{q}S}(C_{1})$ and $\Pi_{T_{q}(S)}(C_{2})$. Moreover $\Pi_{T_{q}(S)}C(q)$ ($C(q)$ here is the orbit passing through $q$) encloses either $\Pi_{T_{q}S}(C_{1})$ or $\Pi_{T_{q}(S)}(C_{2})$. Say, for concreteness, that it encloses $\Pi_{T_{q}(S)}(C_{2})$. Observe also that ${\rm length}(\Pi_{T_{q}S}(C(q)))\leq {\rm length}(C(q))\leq 2\pi L$ where the last inequality is because $q\in B_{\mathbb{R}^{3}}(o,L)$.  
Let $\alpha$ be the curve in ${\mathscr C}$ with constant azimuthal angle $\varphi$ and joining $q$ to $C_{2}$ and observe that $\Pi_{T_{q}S}(\alpha)$ is a straight segment joining $q$ to $\Pi_{T_{q}S}(C_{2})$. Then because $|\nabla u|\leq 1$ we have ${\rm length}(\alpha)\leq 2{\rm length}(\Pi_{T_{q}S}(\alpha))\leq 2 L$ where the last inequality is due to fact that the length of a straight segment inside an ellipse with perimeter less than $2\pi L$ has length less than $L$. But on the other hand $q\in B_{\mathbb{R}^{3}}(o,L)$ and $\partial {\mathscr C}\subset \partial B_{\mathbb{R}^{3}}(o,L_{0}/4)$ and therefore ${\rm length}(\alpha)\geq L_{0}/4-L$. These two inequalities for the length of $\alpha$ are incompatible if $L<L_{0}/20$. 
Hence there are no stable minimal surfaces $S$ in the family ${\mathscr S}_{2}$ intersecting $B_{\mathbb{R}^{3}}(o,L)$ if $L<L_{0}/20$. \end{proof}

\vs
The proof of Proposition \ref{PMSF2} is done in the same way as in Proposition \ref{PMSF} and will not be included here. Let us  prove now Corollary \ref{COR3}. 

\vs
\begin{proof}[\bf Proof of Corollary \ref{COR3}.] In several parts of the proof we will use the ``quotient picture" as explained in Section \ref{PT2}. 

{\it Item 1}. By homogeneous regularity there is $\epsilon_{1}(\rho_{0},\mu_{0},\mu_{1})$ such that for any $0<\epsilon\leq \epsilon_{1}$, the set
\ben
(\partial {\mathcal T}(\partial \Sigma,\epsilon))\setminus \partial \Sigma,
\een
is smooth and strictly convex with mean curvature greater or equal than $\mu_{0}/2$. Thus these surfaces act as barriers preventing  the the existence of minimal surfaces $S$ at a distance less than $\epsilon_{1}$ from $\partial \Sigma$.

{\it Item 2}. By homogeneous regularity there is $\epsilon_{4}(\rho_{0},\mu_{0},\mu_{1})>0$ such that for any $0<\epsilon<\epsilon_{4}$  the surface
\ben
(\partial {\mathcal T}_{g}(\poles,\epsilon))\cap (\Sigma\setminus {\mathcal T}_{g}(\partial \Sigma,\epsilon_{1})),
\een
is smooth (with boundary) and of mean curvature greater or equal than one (in the outgoing direction from the axes $\poles$). Note that as $\epsilon\rightarrow 0$ the mean curvature goes to infinity ($\sim 1/\epsilon$). In particular, because these surfaces act as barriers there are no minimal surfaces lying entirely inside 
\ben
{\mathcal T}_{g}(\poles,\epsilon_{4})\cap (\Sigma\setminus {\mathcal T}_{g}(\partial \Sigma,\epsilon_{1})).
\een 
We will work now in the quotient and follow the notation of Section \ref{PT2}. By homogeneous regularity again, there is $R(\rho_{0},\epsilon_{1},\epsilon_{4})>0$ such that for any $\tilde{p}$
in $\tilde{\Sigma}\setminus \Pi \big({\mathcal T}_{g}(\poles,\epsilon_{4}) \cup {\mathcal T}_{g}(\partial \Sigma,\epsilon_{1})\big)$ we have
\ben
B_{\bar{\tilde{g}}}(\tilde{p},R)\subset \bigg(\tilde{\Sigma}\setminus \Pi \big({\mathcal T}_{g}(\poles,\frac{\epsilon_{4}}{2})\cup {\mathcal T}_{g}(\partial \Sigma,\frac{\epsilon_{1}}{2})\big)\bigg),
\een
and moreover the metric $\bar{\tilde{g}}/R^{2}$ in $B_{\bar{\tilde{g}}/R^{2}}(\tilde{p},1)=B_{\bar{\tilde{g}}}(\tilde{p},R)$ is sufficiently close in $C^{2}$ to the flat metric in $\mathbb{R}^{2}$ that every $\bar{\tilde{g}}/R^{2}$-geodesic (or, the same, every $\bar{\tilde{g}}$-geodesic) passing through $\tilde{p}$ reaches the boundary of $B_{\bar{\tilde{g}}/R^{2}}(\tilde{p},1)$ and therefore has $\bar{\tilde{g}}/R^{2}$-length greater than $1$ (or, the same, the $\bar{\tilde{g}}$-length is greater or equal than $R$). Now, for every axisymmetric minimal surface $S$, $\Pi(S)$ is a $\bar{\tilde{g}}$-geodesic and $A(S)={\rm length}_{\bar{\tilde{g}}}(\Pi(S))$. Moreover by what was said before there is always a point $\tilde{p}$ of $\Pi(S)$ in $\tilde{\Sigma}\setminus \Pi \big({\mathcal T}_{g}(\poles,\epsilon_{4}) \cup {\mathcal T}_{g}(\partial \Sigma,\epsilon_{1})\big)$. It follows that $A(S)\geq R$. The {\it item 2} follows by defining $A_{0}:=R$.  

{\it Item 3}. By homogeneous regularity there exists $R_{1}(\rho_{0},\epsilon_{1},\epsilon_{4})>0$ such that for any $p\in \poles\setminus {\mathcal T}_{g}(\partial \Sigma,\epsilon_{1}/2)$ the metric $g/R_{1}^{2}$ in $B_{g/R_{1}^{2}}(p,2)$ is $\epsilon_{0}$-close in $C^{2}$ to the flat metric in $\mathbb{R}^{3}$ where $\epsilon_{0}$ is as in Proposition \ref{PMSF2}. Then we define $\epsilon_{2}:=LR_{1}$ where $L$ is as in Proposition \ref{PMSF2}. We show now that with this $\epsilon_{2}$ we have all the properties that we desire for {\it item 3}. The cylinder (with boundary)
\ben
\big( \partial {\mathcal T}_{g}(\poles,\epsilon_{2})\big) \cap B_{g}(p,R_{1})=\big(\partial {\mathcal T}_{g/R_{1}^{2}}(\poles, L)\big)\cap B_{g/R_{1}^{2}}(p,1), 
\een
is foliated by $U(1)$-orbits (circles) and for every one, one can consider the area minimizing disc according to Proposition \ref{PMSF2} and with boundary the orbit. This construction can be done for every point $p$ in $\poles\setminus {\mathcal T}_{g}(\partial \Sigma,\epsilon_{1}/2)$ which gives us the set of all the discs we are looking for. Now, let $S$ be a stable and axisymmetric minimal surface embedded in $\Sigma$. If $S\cap {\mathcal T}_{g}(\poles,\epsilon_{2})\neq \emptyset$ then there is $p\in \poles$ such that 
\ben
S\cap B_{g/R_{1}^{2}}(p,L)\neq \emptyset.
\een
But then by Proposition \ref{PMSF2}, $S\cap B_{g/R_{1}^{2}}(p,L)$ must be one of the discs we defined before. Finally note that if an axisymmetric compact and boundary-less surface intersects the axes $\poles$ then it must do twice and the surface must be a sphere.    

{\it Item 4}. By standard minimal surfaces estimates \cite{MR2780140} the Gaussian curvature ${\mathcal K}$ of any disc $D$ as in {\it item 4} and with respect to the induced metric from $g/R_{1}^{2}$ is uniformly bounded. Consider now polar coordinates in $D$ as defined in Section \ref{APCO}. Then, by Gauss-Bonet we have
\ben
2\pi- \frac{d\, l(s)}{d\, s}= \int_{0}^{s}{\mathcal K}(\bar{s}) l(\bar{s}) d\bar{s}.  
\een
From this, the estimate $|l(s)-2\pi s|\leq c_{0}(\rho_{0},\mu_{0},\mu_{1}) s^{2}$ easily follows. 
\end{proof}

\begin{proof}[\bf Proof of Proposition \ref{SHGE}.] Let $i$ be $1$ or $2$. Let $\{S_{m}\}$ be a sequence of axisymmetric spheres in $\Sigma_{i}$. Denote by $\Omega_{m}$ the (closed) set enclosed by $S_{H}$ and $S_{m}$. Assume that the sequence is such that
$\Omega_{m}\subset \Omega_{m+1}$ and ${\rm dist}_{g}(S_{H},S_{m})\rightarrow \infty$.
Let $\{\delta_{m}\}$ be a positive sequence in such a way that (for every $m$) $\delta_{m}$ is sufficiently small that we have $g$-Gaussian coordinates well defined inside ${\mathcal T}_{g}(S_{m},\delta_{m})\setminus \Omega_{m}$. Thus on ${\mathcal T}_{g}(S_{m},\delta_{m})\setminus \Omega_{m}$ we can write
\ben
g=dr^{2}+h_{m},\ 
\een
where $r(p)=dist(p,S_{m})>0$ and $h_{m}(\partial_{r},-)=h_{m}(-,\partial_{r})=0$. Define $\Omega_{\delta_{m}}=\Omega_{m}\cup {\mathcal T}_{g}(S_{m},\delta_{m})$. On $\Omega_{\delta_{m}}$ we define the axisymmetric metric $g^{*}_{m}$ 
\be\label{GSD}
g^{*}_{m}=
\left\{
\begin{array}{lll}
g & {\rm on} & \Omega_{m},\\
dr^{2}+f(r)^{2}h_{m} & {\rm on} & {\mathcal T}_{g}(S_{m},\delta_{m})\setminus \Omega_{m}
\end{array}
\right.
\ee
and where $f(r)$ is the scalar function 
\ben
f(r)=1+e^{\scb{1.1}{\ $\big[\frac{1}{\delta_{m}+\epsilon_{m}-r}-\frac{1}{r}\big]$}},
\een
where $\epsilon_{m}$ is, for every $m$, small enough to make the boundary $\{r=\delta_{m}\}$ strictly convex. Observe that  $g^{*}_{m}=g$ on $\Omega_{m}$ and that $g^{*}_{m}\geq g$ on $\Omega_{\delta_{m}}$ (because $f\geq 1$).  

We define $\underline{A}_{i,m}^{\north}(\tilde{p})$, in the same way as $\underline{A}^{\north}_{i}(\tilde{p})$, as the infimum of the $g^{*}_{m}$-areas of the axisymmetric discs $D^{\north}(C)$, $C=\Pi^{-1}(\tilde{p})$, inside $\Omega_{\delta_{m}}$. 

\vs
\n {\it Claim 1: $\underline{A}^{\north}_{i,m}(\tilde{p})$ is realized by either
\begin{enumerate}[labelindent=\parindent,leftmargin=*,label={\bf D\arabic*{'}}]
\item The $g^{*}_{m}$-area of a disc $D^{\north}(C)=\Pi^{-1}(\gamma^{\north}(\tilde{p}))\subset \Omega_{\delta_{m}}$, or by,
\item The $g^{*}_{m}$-area of a disc $D^{\south}(C)=\Pi^{-1}(\gamma^{\south}(\tilde{p}))\subset \Omega_{\delta_{m}}$ plus $A_{H}$.
\end{enumerate}
}

\n The proof of the {\it claim 1} is as follows. Let $\{D^{\north}_{j}(C)\}=\{\Pi^{-1}(\beta_{j}^{\north}(\tilde{p}))\}$ be a $g^{*}_{m}$-area minimizing sequence of discs in $\Omega_{\delta_{m}}$, that is, a sequence for which we have
\ben
\lim A_{g^{*}_{m}}(D_{j}^{\north}(C))=\lim {\rm length}_{\bar{\tilde{g}}^{*}_{m}} (\beta^{\north}(\tilde{p}))= \underline{A}_{i,m}^{\north}(\tilde{p}),
\een    
(where as we defined $\bar{\tilde{g}}$ we define here $\bar{\tilde{g}}^{*}_{m}:={\lambda^{*}}^{2} \tilde{g}^{*}_{m}$, ${\lambda^{*}}^{2}=<\xi^{*}_{m},\xi^{*}_{m}>_{g^{*}_{m}}$ the $g^{*}_{m}$-norm squared of the rotational Killing $\xi^{*}_{m}$ of the metric $g^{*}_{m}$). 
On general grounds [\footnote{Despite it naturalness this  does not follow exactly from the well known result in Riemannian geometry that a minimizing sequence of curves in a complete boundary-less Riemannian manifold with fixed extreme points converges in measure to a geodesic because on one side the manifold $\tilde{\Sigma}$ has boundary and on the other hand the metric $\bar{\tilde{g}}_{m}^{*}$ is singular at the axis $\poles$ (the boundary of $\tilde{\Sigma}$. Although with more work a proof can be given along these lines, a rigorous proof follows from the standard results on geometric measure theory on area minimizing sequences of discs \cite{MR756417},\cite{MR678484}.}] the sequence of area minimizing discs converges in measure to a disc with boundary $C$ (the solution of the ``Plateau's problem") and, possibly, to a finite set of axisymmetric compact and non-contractible stable minimal surfaces, which because of the geometry of $\Omega_{\delta_{m}}$ must be axisymmetric spheres. Thus the limit is either 
\begin{enumerate}[labelindent=\parindent,leftmargin=*,label={\bf P\arabic*.}]
\item A stable minimal disc $D^{\north}(C)=\Pi^{-1}(\gamma^{\north}(\tilde{p}))$, or,
\item A stable minimal disc $D^{\south}(C)=\Pi^{-1}(\gamma^{\south}(\tilde{p}))$, or,
\item A stable minimal disc $D^{\north}(C)=\Pi^{-1}(\gamma^{\north}(\tilde{p}))$ plus a set of stable axisymmetric minimal spheres $S_{k}=\Pi^{-1}(\gamma_{k}^{\north,\south}),\ k=1,\ldots,k_{1}$, or, 
\item A stable minimal disc $D^{\south}(C)=\Pi^{-1}(\gamma^{\south}(\tilde{p}))$ plus a set of stable axisymmetric minimal spheres $S_{k}=\Pi^{-1}(\gamma_{k}^{\north,\south}),\ k=1,\ldots,k_{2}$.
\end{enumerate} 
Of course to guarantee the existence of the limit one uses that $(\Omega_{\delta_{m}},g^{*}_{m})$ has one totally geodesic boundary ($S_{H}$) and one strictly convex boundary.  
We show now that the cases {\bf P2} and {\bf P3} cannot occur and that case {\bf P3} could occur but when does then it does only with $k_{2}=1$ and $A_{g^{*}_{m}}(S_{1})={\rm length}_{\bar{\tilde{g}}_{m}^{*}}(\gamma^{\north,\south}_{1}(\tilde{p}))=A_{H}$. We show this below, which completes the proof of the {\it claim 1}.

Case {\bf P3} cannot occur for in that case the (constant) sequence of curves $\{\beta{'}^{\north}_{j}(\tilde{p})=\gamma^{\north}(\tilde{p})\}$ has 
\ben
\lim {\rm length}_{\bar{\tilde{g}}_{m}^{*}}(\beta{'}^{\north}_{j}(\tilde{p}))<\lim {\rm length}_{\bar{\tilde{g}}_{m}^{*}}(\beta^{\north}_{j}(\tilde{p})),
\een
which is not possible because $\{\beta^{\north}_{j}(\tilde{p})\}$ is by assumption a minimizing sequence. 

We show now that the case {\bf P2} cannot occur. Let $\epsilon>0$ be small enough such that the curve $\gamma^{\south}(\tilde{p})$ does not intersect ${\mathcal T}_{\bar{\tilde{g}}_{m}^{*}}({\mathcal N},\epsilon)$ where the tubular neighborhood is inside $\Omega^{*}_{m}$. If the sequence $\{\beta^{\north}_{j}(\tilde{p})\}$ converges in measure to $\gamma^{\north}(\tilde{p})$ then
it must be
\ben
\lim {\rm length}_{\bar{\tilde{g}}_{m}^{*}}(\beta^{\north}_{j}(\tilde{p})\cap {\mathcal T}_{\bar{\tilde{g}}_{m}^{*}}({\mathcal N},\epsilon))=0.
\een
But every curve from $\tilde{p}$ to ${\mathcal N}$ must intersect the tubular neighborhood in a curve (or a set of curves) of total length at least $\epsilon$. This gives a contradiction. 

We finally analyze case {\bf P4}. First we show that $k_{2}=1$. Assume by contradiction that $k_{2}\geq 2$. Then one can consider a sequence of curves $\{\beta{'}^{\north}_{j}(\tilde{p})\}$ that would allow us to show that the sequence $\{\beta^{\north}_{j}(\tilde{p})\}$ was not minimizing. 
The sequence $\{\beta{'}^{\north}_{j}(\tilde{p})\}$ is constructed as follows. Let $\{q^{1}_{j}\}$ and $\{q^{2}_{j}\}$ be sequences of points in $\gamma^{\south}(\tilde{p})$ and $\gamma^{\north,\south}_{1}$ respectively converging to points in the south axis ${\mathcal S}$ (i.e. the (south) end points of the curves $\gamma^{\north}(\tilde{p})$ and $\gamma^{\north,\south}_{1}$ respectively).  Then $\beta{'}^{\north}_{j}(\tilde{p})$ is defined (starting from $\tilde{p}$) as $\gamma^{\north}(\tilde{p})$ until $q^{1}_{j}$, then as a curve joining $q^{1}_{j}$ and $q^{2}_{j}$ going very near the south axis ${\mathcal S}$ and then as the piece of $\gamma^{\north,\south}_{1}$ that starts at $q^{2}_{j}$ and end at the north point of it. The curve between $q^{1}_{j}$ and $q^{2}_{j}$ has to be chosen in such a way that its $\bar{\tilde{g}}^{*}_{m}$-length goes to zero as $j\rightarrow \infty$ (note that the pre-image of such curve under $\Pi$ becomes a very thin tube joining $\Pi^{-1}(q^{1}_{j})$ and $\Pi^{-1}(q^{2}_{j})$ of small $g^{*}_{m}$-area). With this definition of $\{\beta{'}^{\north}_{j}(\tilde{p})\}$ we have
\ben
\lim {\rm length}_{\bar{\tilde{g}}_{m}^{*}}(\beta{'}^{\north}_{j}(\tilde{p}))={\rm length}_{\bar{\tilde{g}}_{m}^{*}}(\gamma^{\south}(\tilde{p}))+{\rm length}_{\bar{\tilde{g}}_{m}^{*}}(\gamma^{\north,\south}_{1})<\lim {\rm length}_{\bar{\tilde{g}}_{m}^{*}}(\beta^{\north}_{j}(\tilde{p})),
\een
which is impossible because by assumption $\{\beta^{\north}_{j}(\tilde{p})\}$ was a minimizing sequence. That ${\rm length}_{\bar{\tilde{g}}_{m}^{*}}(\gamma_{1}^{\north,\south})=A_{H}$ is a consequence of the fact that one can easily construct a sequence of curves $\{\beta{'}^{\north}(\tilde{p})\}$ (following a similar procedure as before) converging in measure to $\gamma^{\south}(\tilde{p})\cup \gamma_{H}^{\north,\south}$. 

\vs
\n {\it Claim2: For every $\tilde{p} \in \tilde{\Sigma}_{i}^{\circ}$ and constant $B>0$ there is $m(B)$ such that for any $m\geq m(B)$ and $g^{*}_{m}$-stable axisymmetric minimal disc $D_{m}=D_{m}^{\north}(\Pi^{-1}(\tilde{p}))=\Pi^{-1}(\gamma_{m}^{\north}(\tilde{p}))$ or disc $D_{m}=D_{m}^{\mathcal S}(\Pi^{-1}(\tilde{p}))=\Pi^{-1}(\gamma_{m}^{\mathcal S}(\tilde{p}))$, inside $\Omega_{\delta_{m}}$ and intersecting $S_{m(B)}$ we have}
\ben
A_{g}(D_{m}\cap \Omega_{m(B)})=A_{g^{*}_{m}}(D_{m}\cap \Omega_{m(B)})\geq B.
\een
The Proof of the {\it claim 2} is as follows. Let $m$ and $m(B)$ with $m\geq m(B)$ be arbitrary and let $D_{m}$ be a disc as in the hypothesis. By Corollary \ref{COR3}, there is $\epsilon_{2}>0$ such that if $D_{m}$ intersects ${\mathcal T}_{g}(\poles,\epsilon_{2})\setminus {\mathcal T}_{g}(S_{m(B)},1)$ then it intersects ${\mathcal T}_{g}(\poles, \epsilon_{2})\cap \Omega_{m(B)}$ exactly in a small disc. Therefore if $D_{m}\cap S_{m(B)}\neq \emptyset$ then $\gamma_{m}^{\north}(\tilde{p})$ or $\gamma_{m}^{\mathcal S}(\tilde{p})$ (depending on the case) remains inside $\big(\tilde{\Sigma}_{i}\setminus \Pi({\mathcal T}_{g}(\poles,\epsilon_{2}))\big)\cap \Pi(\Omega_{m(B)})$ until entering $\Pi({\mathcal T}_{g}(S_{m(B)},1))$ for the first time. By the homogeneous regularity of the metric $\bar{\tilde{g}}$ on $\tilde{\Sigma}\setminus \Pi({\mathcal T}_{g}(\poles,\epsilon_{2})$ we conclude that if $m(B)\rightarrow \infty$ then necessarily ${\rm length}_{\bar{\tilde{g}}}(\gamma_{m}^{\north}(\tilde{p})\cap \Pi(\Omega_{m(B)}))\rightarrow \infty$. The claim follows then from the identity $A(\Pi^{-1}(\gamma_{m}^{\north}(\tilde{p}))\cap \Omega_{m(B)})={\rm length}_{\bar{\tilde{g}}}(\gamma_{m}^{\north}(\tilde{p})\cap \Pi(\Omega_{m(B)}))$.

\vs
Now, let $\bar{D}^{\north}(C)$, $\Pi(C)=\tilde{p}$, be a fixed disc in $\Sigma_{i}$. Let $B=2A(\bar{D}^{\north}(C))$ and $m(B)$ as in the {\it claim} (assume that $m(B)$ is big enough that $\bar{D}^{\north}(C)\subset \Omega_{m(B)}$). Then for any $m\geq m(B)$ we have
\ben
\underline{A}^{\north}_{i,m}(\tilde{p})\leq A(\bar{D}^{\north}(C))=\frac{B}{2}.
\een
It follows from the {\it claim 2} that for any $m\geq m(B)$ the minimizers $D^{\north}(C)$ (in case {\bf D1'}) or $D^{\south}(C)$ (in case {\bf D2'}) realizing $\underline{A}^{\north}_{i,m}(\tilde{p})$, lie in $\Omega_{m(B)}$. Therefore for any $m\geq m(B)$ we have
\ben
\underline{A}^{\north}_{i,m}(\tilde{p})=\underline{A}^{\north}_{i,m(B)}(\tilde{p}).
\een
As $\lim_{m\rightarrow \infty} \underline{A}^{\north}_{i,m}(\tilde{p})=\underline{A}^{\north}_{i}(\tilde{p})$ we conclude that the minimizers $D^{\north}(C)$ (in case {\bf D1'}) or $D^{\south}(C)$ (in case {\bf D2'}) realizing $\underline{A}^{\north}_{i,m}(\tilde{p})$ are the minimizers claimed {\bf D1} or {\bf D2} in the statement of the Proposition. The rest of the claim in the Proposition are automatic. 
\end{proof}

\bibliographystyle{plain}
\bibliography{KTS3.bbl}

\begin{thebibliography}{10}

\bibitem{2011CQGra..28j5014A}
A.~{Ace{\~n}a}, S.~{Dain}, and M.~E. {Gabach Cl{\'e}ment}.
\newblock {Horizon area-angular momentum inequality for a class of axially
  symmetric black holes}.
\newblock {\em Classical and Quantum Gravity}, 28(10):105014, May 2011.

\bibitem{2009JHEP...09..044A}
A.~J. {Amsel}, G.~T. {Horowitz}, D.~{Marolf}, and M.~M. {Roberts}.
\newblock {No dynamics in the extremal Kerr throat}.
\newblock {\em Journal of High Energy Physics}, 9:44, September 2009.

\bibitem{MR2884392}
Lars Andersson, Michael Eichmair, and Jan Metzger.
\newblock Jang's equation and its applications to marginally trapped surfaces.
\newblock In {\em Complex analysis and dynamical systems {IV}. {P}art 2},
  volume 554 of {\em Contemp. Math.}, pages 13--45. Amer. Math. Soc.,
  Providence, RI, 2011.

\bibitem{Andersson:2007fh}
Lars Andersson, Marc Mars, and Walter Simon.
\newblock {Stability of marginally outer trapped surfaces and existence of
  marginally outer trapped tubes}.
\newblock {\em Adv.Theor.Math.Phys.}, 12, 2008.

\bibitem{MR0092067}
N.~Aronszajn.
\newblock A unique continuation theorem for solutions of elliptic partial
  differential equations or inequalities of second order.
\newblock {\em J. Math. Pures Appl. (9)}, 36:235--249, 1957.

\bibitem{MR583716}
Yvonne Choquet-Bruhat and James~W. York, Jr.
\newblock The {C}auchy problem.
\newblock In {\em General relativity and gravitation, {V}ol. 1}, pages 99--172.
  Plenum, New York, 1980.

\bibitem{Chrusciel:2007dd}
Piotr~T. Chrusciel.
\newblock {Mass and angular-momentum inequalities for axi-symmetric initial
  data sets. I. Positivity of mass}.
\newblock {\em Annals Phys.}, 323:2566--2590, 2008.

\bibitem{MR2780140}
Tobias~Holck Colding and William~P. Minicozzi, II.
\newblock {\em A course in minimal surfaces}, volume 121 of {\em Graduate
  Studies in Mathematics}.
\newblock American Mathematical Society, Providence, RI, 2011.

\bibitem{2011PhRvL.107e1101D}
S.~{Dain} and M.~{Reiris}.
\newblock {Area-Angular-Momentum Inequality for Axisymmetric Black Holes}.
\newblock {\em Physical Review Letters}, 107(5):051101, July 2011.

\bibitem{2011PhRvD..84l1503J}
J.~L. {Jaramillo}, M.~{Reiris}, and S.~{Dain}.
\newblock {Black hole area-angular-momentum inequality in nonvacuum
  spacetimes}.
\newblock {\em Physical Review Letters D}, 84(12):121503, December 2011.

\bibitem{MR678484}
William Meeks, III, Leon Simon, and Shing~Tung Yau.
\newblock Embedded minimal surfaces, exotic spheres, and manifolds with
  positive {R}icci curvature.
\newblock {\em Ann. of Math. (2)}, 116(3):621--659, 1982.

\bibitem{MR756417}
Leon Simon.
\newblock {\em Lectures on geometric measure theory}, volume~3 of {\em
  Proceedings of the Centre for Mathematical Analysis, Australian National
  University}.
\newblock Australian National University Centre for Mathematical Analysis,
  Canberra, 1983.

\bibitem{MR757180}
Robert~M. Wald.
\newblock {\em General relativity}.
\newblock University of Chicago Press, Chicago, IL, 1984.

\end{thebibliography}

\end{document}